\documentclass[journal,11pt, draftclsnofoot, onecolumn]{IEEEtran}

\IEEEoverridecommandlockouts
\usepackage[cmex10]{mathtools}
\usepackage{mathrsfs}
\usepackage[export]{adjustbox}

\usepackage{cite}
\usepackage{amsmath,amssymb,amsfonts}

\usepackage{mathtools}
\usepackage{algorithmic}
\usepackage{graphicx}
\usepackage{textcomp}
\usepackage{amsthm}
\usepackage{graphicx}
\usepackage{bm}
\usepackage{dsfont}
\usepackage{enumitem}
\usepackage{comment}

\usepackage{subcaption}

\usepackage{enumitem}

\usepackage{color,soul}
\usepackage{multirow,array}
\usepackage{blkarray}
\graphicspath{{./images/}}
\usepackage{amssymb}
\usepackage{amsmath}

\newtheorem{theorem}{Theorem}
\newtheorem{lemma}{Lemma}

\newtheorem{corollary}{Corollary}[theorem]
\newtheorem{corollaryLemma}{Corollary}[lemma]

\newtheorem{definition}{Definition}
\newcommand\mydots{\hbox to 1em{.\hss.\hss.}}

\newtheorem{innercustomgeneric}{\customgenericname}
\providecommand{\customgenericname}{}
\newcommand{\newcustomtheorem}[2]{%
  \newenvironment{#1}[1]
  {%
   \renewcommand\customgenericname{#2}%
   \renewcommand\theinnercustomgeneric{##1}%
   \innercustomgeneric
  }
  {\endinnercustomgeneric}
}
\newcustomtheorem{customclaim}{Claim}

\usepackage[T1]{fontenc}
\usepackage{tabularx,ragged2e,booktabs}

\newcommand*{\set}{\fontfamily{qag}\selectfont}
\DeclareTextFontCommand{\textset}{\set}

\newcommand*\adv{{\tt adv}}

\usepackage{ifthen}
\newboolean{ISIT}
\newboolean{editor}
\begin{document}

\setboolean{ISIT}{false} %Set to true to compile extended version, false to compile ISIT version
\setboolean{editor}{false} %Set to true to compile editor comments, false to hide them

\title{Adversarial Channels with O(1)-Bit Partial Feedback \\
\thanks{This work is supported in part by the U.S National Science Foundation under Grants CCF-1908308, CNS-2212565, CNS-2225577, ITE-2226447 and EEC-1941529 and in part by the Office of Naval Research. A preliminary version of this work will be presented at the 2023 IEEE International Symposium on Information Theory \cite{Ruzomberka2023}.

E. Ruzomberka and H. V. Poor are with the Department of Electrical and Computer Engineering, Princeton University, USA (email: $\{$er6214,poor$\}$@princeton.edu). Y. Jang and D. J. Love are with the Elmore Family School of Electrical and Computer Engineering, Purdue University, West Lafayette, USA (email: $\{$jang216,djlove$\}$@purdue.edu).}
}

\author{\IEEEauthorblockN{Eric Ruzomberka, Yongkyu Jang, David J. Love and H. Vincent Poor}
}

\maketitle

\begin{abstract}
We consider point-to-point communication over $q$-ary adversarial channels with partial noiseless feedback. In this setting, a sender Alice transmits $n$ symbols from a $q$-ary alphabet over a noisy forward channel to a receiver Bob, while Bob sends feedback to Alice over a noiseless reverse channel. In the forward channel, an adversary can inject both symbol errors and erasures up to an error fraction $p \in [0,1]$ and erasure fraction $r \in [0,1]$, respectively. In the reverse channel, Bob's feedback is \textit{partial} such that he can send at most $B(n) \geq 0$ bits during the communication session. 

As a case study on minimal partial feedback, we initiate the study of the $O(1)$-bit feedback setting in which $B$ is $O(1)$ in $n$. As our main result, we provide a tight characterization of zero-error capacity under $O(1)$-bit feedback for all $q \geq 2$, $p \in [0,1]$ and $r \in [0,1]$, which we prove this result via novel achievability and converse schemes inspired by recent studies of causal adversarial channels without feedback. Perhaps surprisingly, we show that $O(1)$-bits of feedback are sufficient to achieve the zero-error capacity of the $q$-ary adversarial error channel with full feedback when the error fraction $p$ is sufficiently small.

\end{abstract}


  \ifthenelse{\boolean{editor}}{
  { \color{red} \textbf{Things to do.}

  \begin{itemize}
  \item Add a diagram of the achievability scheme and converse attack.
  \item Polish the abstract.
  \item Format into 1 column for reviews.
  \item We need overviews (especially for achievability scheme).
  \item Add related work.
  \end{itemize}
  
  }}{}

  \section{Introduction} \label{sec:intro}

One of the oldest questions in coding theory is, ``What is the impact of transmitter feedback on the fundamental limits of reliable communication?" Shannon addressed this question in his 1956 paper \cite{Shannon1956}, in which he showed that feedback does not increase channel capacity for a point-to-point memoryless channel. In the same work, Shannon conversely showed that feedback can increase the so called \textit{zero-error capacity} for certain channels. In the zero-error capacity problem, the focus is on error-correction codes that can be decoded with zero probability of decoding error. Equivalently, one can focus on codes that are robust to noise patterns generated in a \textit{worst-case} manner as if the noise is designed by a malicious adversary who seeks to induce a decoding error.

%In this work, we study the zero-error capacity problem for $q$-ary additive noise channels in which a) an adversary can induce a fraction of symbol errors and a fraction of symbol erasures up to a bound $p \in [0,1]$ and $r \in [0,1]$, respectively. 
% producing a number of capacity characterizations along with constructive coding schemes which can achieve capacity with remarkable simplicity
Consider the communication setting depicted in Fig. \ref{fig:channel_model}, where a sender Alice wishes to communicate a message $m$ from a set $\mathcal{M}$ to a receiver Bob. To send her message, Alice transmits a sequence of $n$ symbols from a $q$-ary input alphabet $\mathcal{X} = \{0,1,\ldots,q-1 \}$ for $q \geq 2$ over an additive noise channel. The channel is controlled by an adversary, who, for an error fraction $p \in [0,1]$ and erasure fraction $r \in [0,1]$, can induce up to $pn$ symbol errors and up to $rn$ symbol erasures in Alice's transmission.

To assist Bob in correcting these errors and erasures, Alice may adapt her sequence during transmission using feedback received from Bob. Over the entire communication session, Bob sends $T(n)$ feedback symbols, each from a feedback alphabet $\mathcal{Z}$, which Alice receives causally without delay or noise. Critically, Bob's noiseless feedback is \textit{partial} such that he can send at most $B(n)$ bits of feedback over the session, i.e., $T(n) \log_2 |\mathcal{Z}| \leq B(n)$, for some function $B:\mathbb{Z}^+ \rightarrow \mathbb{Z}^+$ defined on the non-negative integers $\mathbb{Z}^+$. As extreme cases, the functions $B(n) = (n-1) \log_2(q+1)$ and $B(n) = 0$ correspond to full noiseless feedback\footnote{More specifically, full feedback corresponds to $T(n) = n-1$ feedback symbols from a feedback set $\mathcal{Z}$ of size $2^{B(n)/T(n)} = (q+1)$, i.e., Bob can report one $q$-ary symbol or erasure symbol of feedback per Alice's transmission.} and no feedback, respectively. For general $B$, Bob may not send a feedback symbol after every received symbol. After transmission, Bob attempts to decode by choosing a message estimate $\hat{m} \in \mathcal{M}$, where a decoding error occurs if $\hat{m} \neq m$. We refer to the above model as the $q$-ary adversarial error-erasure channel with $B$-bit feedback. 

\begin{figure}[t]
  \centering
  \includegraphics[width=0.7\columnwidth]{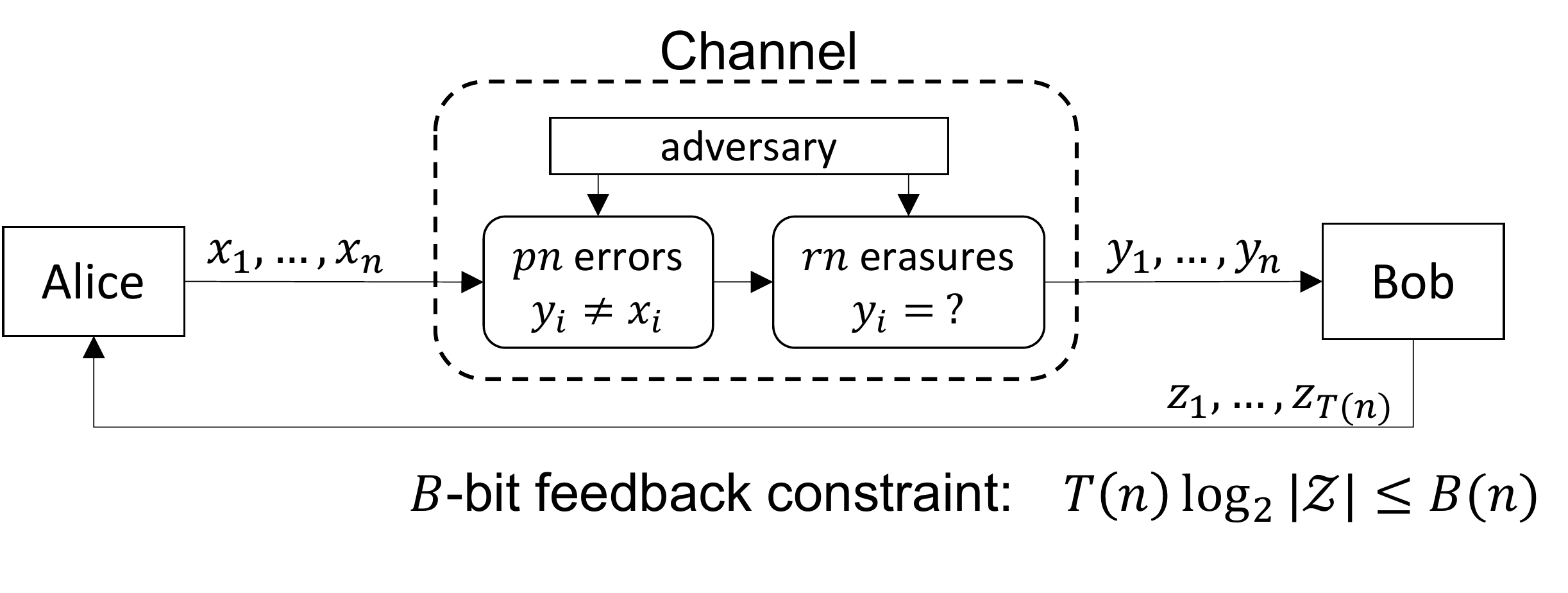}
  \caption{The $q$-ary adversarial error-erasure channel with $B$-bit feedback.}
  \label{fig:channel_model}
\end{figure}

Roughly, a rate $R > 0$ is said to be \textit{(zero-error) achievable with $B$-bit feedback under an error fraction $p$ and erasure fraction $r$} if for large enough $n$ and for message size $|\mathcal{M}| = q^{nR}$ there exists a coding scheme such that no decoding error occurs for any $m \in \mathcal{M}$ and any adversarial choice of error and erasure positions. In turn, the \textit{zero-error capacity} $C^B_q(p,r)$ of the $q$-ary adversarial error-erasure channel with $B$-bit feedback is the supremum of rates achievable with $B$-bit feedback under and error fraction $p$ and erasure fraction $r$. 

%For rate $R = \frac{1}{n} \log_q |\mathcal{M}|$ and $B \geq 0$ bits of feedback, an $(n,Rn,B)$-code is a scheme that makes $n$ transmissions in the \textit{forward channel} (i.e., from Alice to Bob) and at most $B$ transmissions comprising $B$ bits in total in the \textit{reverse channel} (i.e., from Bob to Alice). 

\subsection{Binary Adversarial Error Channels with $B$-bit Feedback}

Since Shannon's work, the zero-error capacity problem has been well-studied for special cases of the $q$-ary adversarial error-erasure channel with $B$-bit feedback. To motivate our work, we focus now on the \textit{binary adversarial error channel} (i.e., $q=2$, $p \in [0,1]$ and $r = 0$) due to its widespread attention in the literature. For the more general $q$-ary error-erasure model, we delay our discussion of known results until Section \ref{sec:results}. 

The earliest investigations into the binary adversarial error channel were conducted by Berlekamp \cite{Berlekamp1964} and later Zigangirov \cite{Zigangirov1976} who fully characterized the zero-error capacity under \textit{full (noiseless) feedback}, i.e., $B(n) = n-1$. Denote the capacity under full feedback as $C^{\mathrm{full}}_2(p,0) \triangleq C^{B}_2(p,0)$ for $B(n) = n-1$. The result of \cite{Berlekamp1964,Zigangirov1976} is

\begin{equation} \label{eq:cap_full_2}
  C^{\mathrm{full}}_2(p,0) = \begin{cases}
  1 - H_2(p), & p \in [0,p^\ddagger] \\
  \frac{1-H_2(p^\ddagger)}{1/3-p^\ddagger}(1/3-p), & p \in (p^\ddagger,1/3) \\
  0, &p \in [1/3,1]
  \end{cases}
  \end{equation}
  where $p^\ddagger = \frac{1}{4}(3-\sqrt{5})$ and where $H_2(\cdot)$ is the binary entropy function. Hence, $C^{\mathrm{full}}_2(p,0)$ is positive for all error fractions $p \in [0,1/3)$. In contrast, the zero-error capacity $C_{2}^{0}(p,0)$ without feedback is zero for all $p > 1/4$ by the Plotkin bound \cite{Plotkin1960}.\footnote{An exact characterization of the zero-error capacity $C^{0}_2(p,0)$ of the binary adversarial error channel without feedback remains an open problem. The best lower and upper bounds are given by the Gilbert-Varshamov bound \cite{Gilbert1952a,Varshamov1957EstimateCodes} and the linear programming bound (or MRRW bound) of McEliece, Rodemich, Rumsey and Welch \cite{McEliece1977a}, respectively.} Furthermore, by the linear programming bound  \cite{McEliece1977a}, feedback strictly increases the zero-error capacity of the binary adversarial error channel for all $p \in (0,1/3)$.

%\footnote{Under full noiseless feedback, a result of Berlekamp \cite{Berlekamp1964} is that no scheme with positive rate is resilient to an error fraction above $1/3$.}

In many practical communication systems, full feedback is impractical due to the high cost of feedback resources (e.g., bandwidth, power). Motivated by this consideration, Haeupler, Kamath and Velingker \cite{Haeupler2015} initiated the study of the \textit{partial (noiseless) feedback} setting in which $B(n) = \delta n$ bits of feedback are sent for some $\delta \in (0,1]$. We note that $\delta = 1$ and $\delta=0$ coincides with full feedback and no feedback, respectively. A result of \cite{Haeupler2015} is that the zero-error capacity for the binary adversarial error channel under partial feedback remains positive for all $p \in [0,1/3)$ when $\delta \in (2/3,1]$.\footnote{This result was extended in \cite{Joshi2022} to larger alphabet channels.} 

More recently, Gupta, Guruswami and Zhang \cite{Gupta2022} significantly improved upon this result and showed that \textit{sub-linear feedback} is sufficient for the capacity to be positive for all $p \in [0,1/3)$. In particular, just $B = O(\log n)$ bits of feedback are sufficient. Conversely, the same authors showed that $B = \Omega(\log n)$ bits of feedback are necessary: for $p > 1/4$ the zero-error capacity is exactly $0$ when $B = o(\log n)$. \textit{Hence, when only $B = o( \log n)$ bits of feedback are available, the support of the zero-error capacity of the binary adversarial error channel with $B$-bit feedback coincides with the support when no feedback is available.}

% In real-world communication systems, even $\Omega(\log n)$ bits of feedback can be untenable considering that feedback is often restricted to a few bits per transmission block, 

\subsection{$q$-ary Adversarial Error-Erasure Channels with $O(1)$-bit Feedback}

In light of this negative result, one may wonder if $B = o(\log n)$ bits of feedback are still useful from a capacity point-of-view. The extent to which feedback is used in real-world communication systems suggests that it may be too costly to require that the number of feedback bits scale in the number of transmitted symbols $n$. In practice, feedback is often restricted to a few bits per transmission block, e.g., in LTE/5G feedback-supported protocols such as hybrid-ARQ, channel precoding for multi-antenna wireless, and CSI usage \cite{Lott2007,Love2008,Ku2014}. 

In this work, we consider a more limited form of partial noiseless feedback than \cite{Haeupler2015,Joshi2022,Gupta2022} in which the number of feedback bits $B(n)$ is $O(1)$ i.e., does \textit{not} grow with the number of transmitted symbols $n$. We consider the general $q$-ary adversarial error-erasure channel with $O(1)$-bit feedback. The zero-error capacity of the $q$-ary adversarial error-erasure channel with $O(1)$-bit feedback is defined as 
\begin{equation} \label{eq:cap_def}
C_q^{O(1)}(p,r) \triangleq \sup_{B \in \mathscr{C}_{\mathrm{const}}} C^{B}_q(p,r)
\end{equation} 
where $\mathscr{C}_{\mathrm{const}} = \{B: \mathbb{Z}^+ \rightarrow \mathbb{Z}^+ | b \in \mathbb{Z}^+, B(n)=b \text{ for all } n\}$.
We remark that the supremum in (\ref{eq:cap_def}) implies that there may not exist a single $b \in \mathbb{Z}^+$ such that every rate $R < C_q^{O(1)}(p,r)$ is achievable by a coding scheme using a finite number of bits $b$. In this work, we show that a rate $R< C_q^{O(1)}(p,r)$ is achievable for some finite number $b(R)$ where $b(R)$ tends to $\infty$ as $R$ tends to $C_q^{O(1)}(p,r)$.

\begin{figure*}

\begin{minipage}[c]{\textwidth}
    \centering
    %\raggedright
    \begin{subfigure}{0.45\textwidth}
    \centering
    \includegraphics[width=\textwidth]{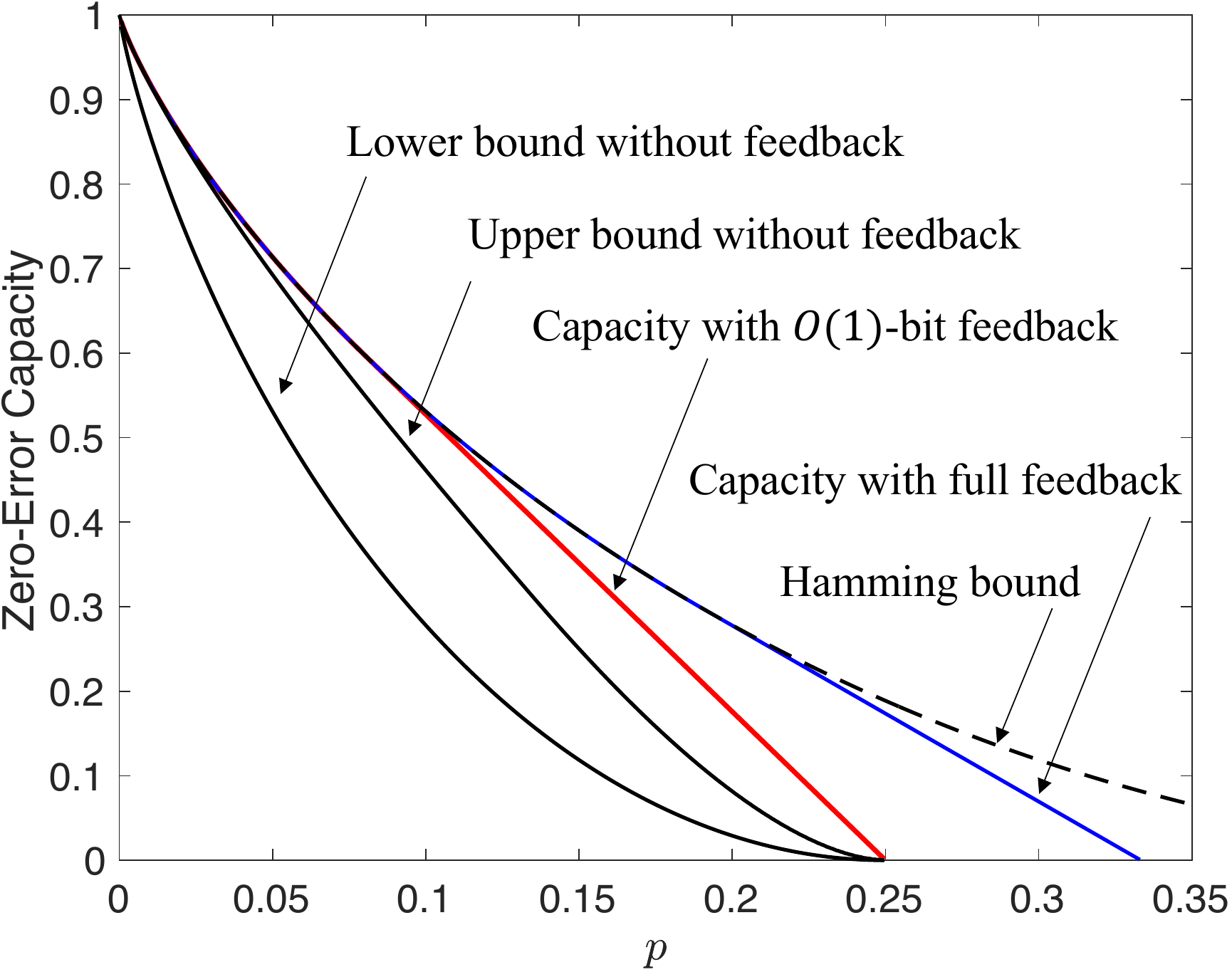}
    \caption{}
    %\vspace{-0.2in}
    \end{subfigure}
    ~
    \begin{subfigure}{0.45\textwidth}
    \centering
    \includegraphics[width=\textwidth]{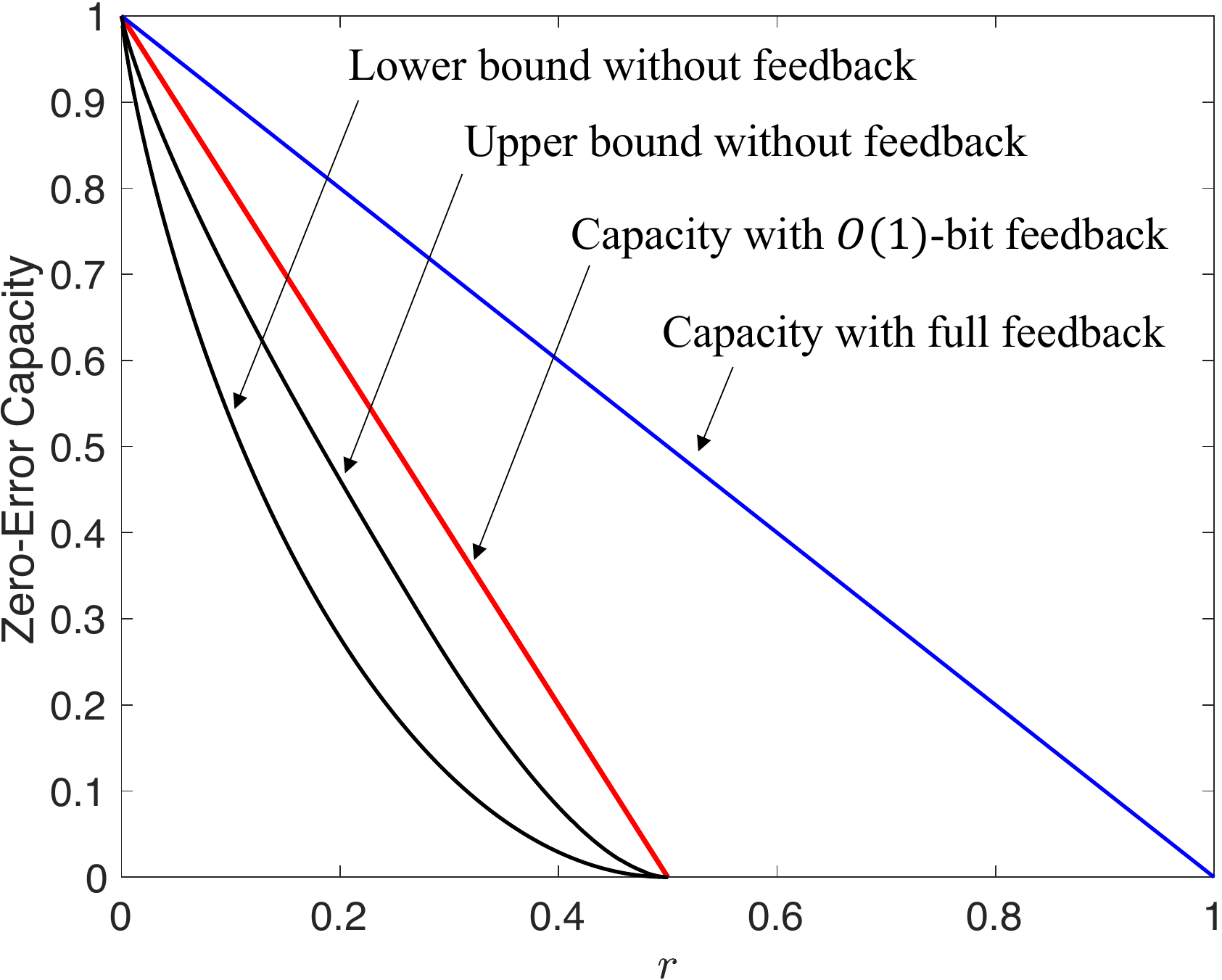}
    \caption{}
    \end{subfigure}
    \end{minipage}\hfill
    \caption{Zero-error capacity bounds of (a) binary adversarial error channels and (b) binary adversarial erasure channels. Red plots show the capacity $C^{O(1)}_2(p,r)$ of binary channels with $O(1)$-bit feedback (Theorem \ref{thm:cap_bf}). The capacity of the channel without feedback has best known lower bounds and upper bounds given by the GV bound \cite{Gilbert1952a,Varshamov1957EstimateCodes} and MRRW bound \cite{McEliece1977a}, respectively. The capacity $C^{\mathrm{full}}_2(p,0)$ of the error channel with full feedback is given by Berlekamp \cite{Berlekamp1964} and Zigangirov \cite{Zigangirov1976}. The dashed plot shows the Hamming bound $1-H_2(p)$.}
    \label{fig:bin_channels}
    %  Binary error channels (a): The capacity of the channel without feedback has best known lower bounds and upper bounds given by the Gilbert-Varshamov bound \cite{Gilbert1952a,Varshamov1957EstimateCodes} and MRRW bound \cite{McEliece1977a}, respectively. The red plot shows the capacity of the channel with $O(1)$-bit feedback (Theorem \ref{thm:cap_bf}). The capacity of the channel with full feedback is given by Berlekamp \cite{Berlekamp1964} and Zigangirov \cite{Zigangirov1976}. The dashed plot shows the Hamming bound $1-H_2(p)$. Binary erasure channels (b): 
  \end{figure*}

  \subsection{Results \& Discussion} \label{sec:results}
  In this work, we study the zero-error capacity $C^{O(1)}_q(p,r)$ for the $q$-ary adversarial error-erasure channel with $O(1)$-bit feedback. As our main result, we present a tight characterization of $C^{O(1)}_q(p,r)$ for all error fractions $p\in [0,1]$, erasure fractions $r \in [0,1]$ and alphabet sizes $q \geq 2$. Our proof of this result involves a novel coding scheme to prove a lower bound on $C^{O(1)}_q(p,r)$ and a converse analysis to prove a matching upper bound on $C^{O(1)}_q(p,r)$, both of which are inspired from prior work \cite{Dey2013,Chen2014,Chen2019} on causal channels without transmitter feedback (see Section \ref{sec:related_work} for a detailed discussion). For $q \geq 2$, let $H_q(x)$ denote the $q$-ary entropy function which is equal to $H_q(x) \triangleq x \log_q (q-1) - x \log_q x - (1-x) \log_q (1-x)$ for $x \in [0,1]$.
  
  \begin{theorem} \label{thm:cap_bf}
  Suppose that $q\geq 2$, $p \in [0,1]$ and $r \in [0,1]$. The zero-error capacity of the $q$-ary adversarial error-erasure channel with $O(1)$-bit feedback is 
  \begin{align} %\label{eq:cap_bf}
   C^{O(1)}_q(p,r) =  \begin{cases}
   \min\limits_{\bar{p} \in [0,p]} \left[ \alpha(\bar{p}) \left(1-H_q\left(\frac{\bar{p}}{\alpha(\bar{p})}\right)\right)
   \right],  & \hspace{-0.7em} 2p+r < \frac{q-1}{q} \\
   0, & \hspace{-0.7em} \text{otherwise}
   \end{cases} \nonumber
  \end{align}
  where $\alpha(\bar{p}) = 1 - \frac{2q}{q-1}(p-\bar{p}) - \frac{q}{q-1}r$.
  \end{theorem} 

  We remark that in our achievability proof of Theorem \ref{thm:cap_bf}, our coding scheme uses a number of feedback bits $B$ that varies with the coding rate. For fixed $p \in [0,1]$, $r \in [0,1]$ and $q \geq 2$, our coding scheme of rate $R < C^{O(1)}_q(p,r)$ uses a number of feedback bits $B$ which is constant in $n$ but tends to infinity as the rate-to-capacity gap $\epsilon_R = C^{O(1)}_q(p,r) - R$ tends to $0$.
  
  To illustrate the capacity expression of Theorem \ref{thm:cap_bf}, we focus on two special cases of parameters: when the channel can only induce errors (i.e., $p \in [0,1]$ and $r = 0$) and when the channel can only induce erasures (i.e., $p=0$ and $r \in [0,1]$). In each of these cases, the capacity expression can be simplified, as shown in Fig. \ref{fig:bin_channels} for binary alphabets. When only erasures occur, it is easy to verify that $C^{O(1)}_q(0,r)$ is a linear piece-wise function of $r$:
  \begin{equation} \nonumber
  C^{O(1)}_q(0,r) = \begin{cases}
  1- \frac{q}{q-1}r, & r \in [0, \frac{q-1}{q}) \\
  0, & r \in [\frac{q-1}{q},1].
  \end{cases}
  \end{equation}
  When only errors occur, and for $p \in [0,\frac{q-1}{2q})$, $C^{O(1)}_q(p,0)$ is equal to the Hamming bound $1-H_q(p)$ for small $p$ and is otherwise equal to the line tangent to $1-H_q(p)$ and which intersects the point $(p,0)$ where $p=\frac{q-1}{2q}$. That is,
  \begin{equation} \label{eq:cap_error}
  C^{O(1)}_q(p,0) = \begin{cases}
  1 - H_q(p), & p \in [0,p^*] \\
  \frac{1-H_q(p^*)}{\frac{q-1}{2q} - p^*}(\frac{q-1}{2q}-p), & p \in (p^*, \frac{q-1}{2q}) \\
  0, & p \in [\frac{q-1}{2q},1]
  \end{cases}
  \end{equation}
  where $p^* = p^*(q)$ is the unique value in $[0,\frac{q-1}{2q}]$ that satisfies equation $p^*(1-p^*)^{\frac{q+1}{q-1}} = (q-1)q^{-\frac{2q}{q-1}}$. The proof of (\ref{eq:cap_error}) is in Appendix \ref{sec:main_cor1_proof}.

  We compare the above result to the zero-error capacity of the $q$-ary adversarial error-erasure channel with full feedback, which we denote as $C^{\mathrm{full}}_q(p,r)$. By definition, it is clear that $C^{\mathrm{full}}_q(p,r)$ is an upper bound of $C^{O(1)}_q(p,r)$. We briefly summarize known characterizations of $C^{\mathrm{full}}_q(p,r)$. For $q=2$, $C^{\mathrm{full}}_q(p,0)$ is given by (\ref{eq:cap_full_2}), which has a two-part form similar to (\ref{eq:cap_error}) comprised of the Hamming bound and its tangent. For $q \geq 3$, a result of Ahlswede, Deppe and Lebedev \cite{Ahlswede2006} is the upper bound 
  \begin{equation} \label{eq:q_ub}
  C^{\mathrm{full}}_q(p,0) \leq 
  \begin{cases}
  1-H_q(p), & p \in [0,\frac{1}{q}] \\
  \frac{1-H_q(1/q)}{1/2-1/q}(1/2-p), & p \in (\frac{1}{q},\frac{1}{2}) \\
  0, & p \in [1/2,1]
  \end{cases}
  \end{equation}
  %\begin{equation} \label{eq:q_ub}
  %C^{\mathrm{full}}_q(p,0) \leq 
  %\begin{cases}
  %1-H_q(p), & p \in [0,\frac{1}{q}] \\
  %(1-2p) \log_q(q-1), & p \in (\frac{1}{q},\frac{1}{2}) \\
  %0, & p \in [1/2,1]
  %\end{cases}
  %\end{equation}
  which is tight for all $p \in [1/q,1]$.\footnote{We remark that the coding scheme proposed in \cite{Ahlswede2006} to achieve $C^{\mathrm{full}}_q(p,0)$ for $p \in [1/q,1/2)$ uses full (noiseless) feedback. The authors of \cite{Ahlswede2006} refer to their scheme as the `rubber method'.} However, for $q \geq 3$ and $p \in [0,1/q]$, a tight characterization of $C^{\mathrm{full}}_q(p,0)$ remains open. A corollary of Theorem \ref{thm:cap_bf} is that that the upper bound (\ref{eq:q_ub}) is tight for small values of $p$.

  \begin{corollary}[Full Feedback] \label{thm:main_cor1}
  Suppose $q \geq 2$. The capacity of the $q$-ary error channel with full feedback is
  \begin{equation} \nonumber
  C^{\mathrm{full}}_q(p,0) = C^{O(1)}_q(p,0) = 1 - H_q(p), \hspace{1em} p \in [0,p^*]
  \end{equation}
  where $p^* \in [0,\frac{q-1}{2q}]$ is the unique value that satisfies equation $p^*(1-p^*)^{\frac{q+1}{q-1}} = (q-1)q^{-\frac{2q}{q-1}}$. 
  \end{corollary}

  Corollary \ref{thm:main_cor1} follows via inspection of (\ref{eq:cap_error}) and (\ref{eq:q_ub}), where we remark that $p^* \leq 1/q$ (c.f. Appendix \ref{sec:p_star_q}). We reiterate that for $q \geq 2$ and for $p \in (0,p^*]$, our results show that the zero-error capacity with full feedback $C^{\mathrm{full}}_q(p,0)$ can be achieved with some coding scheme that uses only $O(1)$-bit feedback. While a scheme using $O(1)$-bit feedback cannot be used to achieve $C^{\mathrm{full}}_q(p,0)$ for $p \in (p^*,1]$, the following corollary implies that such a scheme can achieve rates close to $C^{\mathrm{full}}_q(p,0)$ for all $p \in [0,1]$ when the alphabet size $q$ is large. 

  \begin{corollary}[Large Alphabets] \label{thm:lg_alphabet}
  Suppose that $p \in [0,1]$ and $r \in [0,1]$. If $2p+r<1$ then $$C^{O(1)}_q(p,r) = 1 - 2p - r - \Theta\left(\frac{1}{q}\right) \text{ as } q \rightarrow \infty.$$ Otherwise, if $2p+r \geq 1$ then $C^{O(1)}_q(p,r) =0$ for all $q \geq 2$. Proof is in Appendix \ref{sec:lg_alphabet_proof}.
  \end{corollary}

  We remark that for $2p+r<1$, both letters $C^{O(1)}_q(p,r)$ and $C^{\mathrm{full}}_q(p,r)$ tend to the same limit $1-2p-r$ as $q$ tends to infinity, albeit at different rates. $C^{O(1)}_q(p,r)$ tends to the limit slightly slower than $C^{\mathrm{full}}_q(p,r)= 1 - 2p - r - \Theta\left(\frac{1}{q\log q}\right)$. Conversely, $C^{O(1)}_q(p,r)$ tends to the limit $1-2p-r$ faster than the best known lower bound on the zero-error capacity without feedback $C_q^{0}(p,r) \geq 1-2p-r-\Theta(\frac{1}{\sqrt{q}})$ for $q \geq 49$ and $\sqrt{q} \in \mathbb{Z}^+$ due to Tsfasman, Vl\u{a}duts and Zink \cite{Tsfasman1982}.\footnote{The lower bound of \cite{Tsfasman1982} follows from the study of algebraic geometry codes and only holds for $q \geq 49$ and when $\sqrt{q}$ is an integer. For general $q$, the best known lower bound on the zero-error capacity without feedback is the Gilbert-Varshamov bound \cite{Gilbert1952a,Varshamov1957EstimateCodes}.}
  
  \subsection{Related Work} \label{sec:related_work}

  As discussed above, our study of channels with $O(1)$-bit feedback is related to prior studies on channels with full feedback \cite{Berlekamp1964,Zigangirov1976,Ahlswede2006} and partial feedback \cite{Haeupler2015,Joshi2022,Gupta2022}. More generally, our study is related to \textit{adversarial channels} -- a channel modeling framework in which the channel noise is chosen by a malicious adversary seeking to disrupt communication. Adversarial channels may be modeled with or without feedback. These include myopic channels \cite{Sarwate2010a,Dey2019a, Ruzomberka2022} and causal adversarial channels \cite{Chen2014,Dey2013,Chen2019,Suresh2022,Zhang2022}. An important generalization of the adversarial channel framework are arbitrarily varying channels (AVCs) \cite{Joshi2022,Csiszar1988TheConstraints}.

  % One connection between the two models follows from the observation that feedback imposes a causal structure on encoding and decoding procedures, and thus, it is natural to use causal reasoning when thinking about channel processes that can be arbitrarily correlated to these encoding and decoding procedures.
  
  Among the above adversarial channel models, causal adversarial channels (without transmitter feedback) have a particularly close connection to adversarial error-erasure channels with $O(1)$-bit feedback. A channel is said to be \textit{causal} if the adversary's choice to induce an error-erasure in the $i$th transmitted symbol depends only on previously transmitted symbols, i.e., symbols $1$ through $i-1$. One connection between the two models, which appears at first glance to be coincidental, is that the bounded-error capacity of the $q$-ary error-erasure causal channel \cite{Chen2019} coincides exactly with the zero-error capacity of the $q$-ary adversarial error-erasure channel with $O(1)$-bit feedback (Theorem \ref{thm:cap_bf}). The authenticity of this connection becomes apparent in our proof of Theorem 1, which uses insights and tools developed in the study causal adversarial channels \cite{Dey2013,Chen2019,Chen2014} to prove both a lower bound (achievability) and a tight upper bound (converse) of $C^{O(1)}_q(p,r)$.  

  Our converse proof is based on the so called ``babble-and-push'' adversarial attack of Dey, Jaggi, Langberg, Sarwate and Chen \cite{Dey2013,Chen2019} where it was proposed to study upper bounds on the capacity of causal adversarial channels without transmitter feedback. In the converse analysis, ``babble-and-push'' is an attack strategy used by an adversary to confuse Bob about Alice's transmitted codeword. A key step in the converse analysis is to bound the number of codewords within a specified Hamming distance using the Plotkin bound. We propose a novel extension of the standard ``babble-and-push'' framework to incorporate $O(1)$-bit feedback by using ideas from Ramsey theory. 
  
  Furthermore, our achievability proof uses a novel coding scheme that borrows ideas from the capacity achieving scheme of Chen, Jaggi and Langberg \cite{Chen2014,Chen2019} for causal adversarial channels without feedback. An innovation of \cite{Chen2014,Chen2019} is an iterative decoding process in which Bob first \textit{guesses} the number of symbol error in his received sequence at each time in a set of quantized time steps, followed by a guess refinement process that halts (with high probability) at a time step in which Bob's guess coincides with the true number of symbol errors. In turn, Bob uses this number as side-information to decode Alice's message. Since guess refinement may fail with a small probability, decoding may also fail with a small probability. Our scheme adapts the scheme of \cite{Chen2014,Chen2019} by incorporating feedback to ensure that the guess refinement process succeeds with certainty.  

  \subsection{Organization \& Notation}

  The remainder of the paper is organized as follows. In Section \ref{sec:model}, we provide a detailed channel model of the $q$-ary error-erasure channel with $B$-bit feedback. In Section \ref{sec:scheme}, we present the coding scheme that we use to prove the lower bound of Theorem \ref{thm:cap_bf}. An overview and detailed proof of the lower bound is presented in Section \ref{sec:lb}. In Section \ref{sec:ub}, we present an overview and detailed proof of the upper bound of Theorem \ref{thm:cap_bf}. Lastly, in Section \ref{sec:conclusion}, we provide concluding remarks.

  The following notation is used throughout the paper. For an integer $q \geq 2$, define $\mathcal{Q} = \{0,1,\ldots,q-1\}$. For an integer $n \geq 1$, the notation $[n]$ denotes the set $\{1,2,\ldots,n\}$. For two sequences $\bm{a}, \bm{b} \in \mathcal{Q}^n$, the Hamming distance $d_H(\bm{a},\bm{b})$ between $\bm{a}$ and $\bm{b}$ is  defined as the number of positions $i \in [n]$ in which $a_i \neq b_i$. We extend this definition of Hamming distance to account for sequences containing the erasure symbol `?' by defining the distance $d_H(\bm{a},\bm{b})$ between the sequences $\bm{a} \in \mathcal{Q}^n$ and $\bm{b} \in (\mathcal{Q} \cup \{? \} )^n$ to be the number of positions $i \in [n]$ in which $b_i \neq ?$ and $a_i \neq b_i$. Define the non-negative integers $\mathbb{Z}^+ = \{0,1,2,\ldots\}$. For an event $\mathcal{A}$, let $\mathds{1} \{ \mathcal{A} \}$ denote the indicator of $\mathcal{A}$.

  \section{Channel Model} \label{sec:model}

  \subsection{Channel Model}

  For $q \geq 2$, the channel is characterized by an input alphabet $\mathcal{X} = \mathcal{Q} \triangleq \{0,1,\ldots,q-1\}$, an output alphabet $\mathcal{Y} = \mathcal{Q} \cup \{ ? \}$ for an erasure symbol `$?$', and a channel mapping $\adv: \mathcal{X}^n \rightarrow \mathcal{Y}^n$ that is chosen from a set of mappings $\mathcal{ADV}$ by an \textit{adversary} who seeks to disrupt communication between Alice and Bob. For each mapping $\adv \in \mathcal{ADV}$, the channel constraint requires that the number of erasure symbols `$?$' and the number of symbols in error in the channel output cannot exceed $rn$ and $pn$, respectively. The adversary chooses $\adv \in \mathcal{ADV}$ using knowledge of Alice's message $m$ and the following coding scheme, and $\adv$ is not revealed to either Alice or Bob.
  
  \subsection{Codes with Feedback} \label{sec:cwf}
  For a rate $R \in (0,1]$, blocklength $n \geq 1$, and number of feedback bits $B(n) \geq 0$ for some function $B:\mathbb{Z}^+ \rightarrow \mathbb{Z}^+$, an $(n,Rn,B(n))$-code (with feedback) is a tuple $\Psi = (\mathcal{C}_k,\phi,f_k,\mathcal{T},\mathcal{Z})$ that specifies the following communication scheme. 
  
  First, the code $\Psi$ specifies how Bob sends feedback to Alice. For an integer $T = T(n) \geq 0$, Bob sends feedback in $T$ rounds, sending a symbol from the feedback alphabet $\mathcal{Z}$ at each time in the set $\mathcal{T} = \{t_1,\ldots,t_T\}$ where $1\leq t_1 < t_2 < \cdots < t_T < n$. For $k \in [T]$, the feedback symbol sent at time $t_k$ is determined by the feedback function $f_k:\mathcal{Y}^{t_k} \rightarrow \mathcal{Z}$ and is denoted as $z_k = f_k(y_1,\ldots,y_{t_k})$ where $(y_1,\ldots,y_{t_k}) \in \mathcal{Y}^{t_k}$ is Bob's received sequence up to time $t_k$. Notice that $B(n)$ bits of feedback implies that $T \log_2 |\mathcal{Z}| \leq B(n)$. 
  
  Second, $\Psi$ specifies how Alice transmits to Bob. For $k \in [T+1]$, Alice sends symbols in blocks of symbol size $t_{k}-t_{k-1}$ (where we define $t_0 \triangleq 0$ and $t_{T+1} \triangleq n$) using the encoding function $\mathcal{C}_k: \mathcal{M} \times \mathcal{Z}^{k-1} \rightarrow \mathcal{X}^{t_{k}-t_{k-1}}$. Then Alice's transmitted symbols over this block are $(x_{t_{k-1}+1}, \ldots, x_{t_k}) = \mathcal{C}_k(m;z_1,\ldots,z_{k-1})$. Finally, Bob decodes with the the decoding function $\phi:\mathcal{Y}^n \rightarrow \mathcal{M}$. We assume that adversary knows the code $\Psi$ used by Alice and Bob and can use this knowledge to choose $\adv \in \mathcal{ADV}$.

  Let the symbol $\circ$ denote a concatenation between two sequences, i.e., $\bm{a} \circ \bm{b} = (\bm{a},\bm{b})$. We let $f(\bm{y})$ denote the concatenation of feedback functions $f_1\left(y_1,\ldots,y_{t_1}\right) \circ f_2(y_1,\ldots,y_{t_2}) \circ \cdots \circ f_T(y_1,\ldots,y_{t_T})$ and we let $\mathcal{C}(m;z_1,\ldots,z_{T})$ denote the concatenation of encoding functions $\mathcal{C}_1(m) \circ \mathcal{C}_2(m,z_1) \circ \cdots \circ \mathcal{C}_{T+1}(m;z_1,\ldots,z_T).$ In the sequel, for any message $m \in \mathcal{M}$ and feedback sequence $(z_1,\ldots,z_T) \in \mathcal{Z}^T$ we refer to $$\mathcal{C}(m;z_1,\ldots,z_{T})$$ as the \textit{codeword} corresponding to $m$ and $(z_1,\ldots,z_{T})$. Similarly, for $k \in [T+1]$ we refer to $$\mathcal{C}_k(m;z_1,\ldots,z_{k-1})$$ as the \textit{$k$-th sub-codeword} corresponding to $m$ and $(z_1,\ldots,z_{k-1})$.

  \subsection{Capacity}

  For a function $B: \mathbb{Z}^+ \rightarrow \mathbb{Z}^+$, a rate $R \in (0,1]$ is \textit{(zero-error) achievable with $B$-bit feedback under an error fraction $p$ and erasure fraction $r$} if for all $n$ large enough there exists an $(n,Rn,B(n))$-code $\Psi = (\mathcal{C}_k,\phi,f_k,\mathcal{T},\mathcal{Z})$ such that $\phi\left(\adv(\mathcal{C}(m;\bm{z}))\right) = m$ for all $m \in \mathcal{M}$ and all $\adv \in \mathcal{ADV}$ where $\bm{z} = f(\adv(\mathcal{C}(m;\bm{z})))$. The zero-error capacity $C^{B}_q(p,r)$ is the supremum of rates achievable with $B$-bit feedback under an error fraction $p$ and erasure fraction $r$. In turn, we define $$C^{O(1)}_q(p,r) = \sup_{B \in \mathscr{C}_{\mathrm{const}}} C^B_q(p,r)$$ where $\mathscr{C}_{\mathrm{const}} = \{B: \mathbb{Z}^+ \rightarrow \mathbb{Z}^+ | b \in \mathbb{Z}^+, B(n)=b \text{ for all } n\}$.

  \section{Achievability Scheme} \label{sec:scheme}
  
  To prove the lower bound in Theorem \ref{thm:cap_bf}, we use a specific code construction which we present here. Throughout our construction, we let rate $R\in (0,1]$, blocklength $n \geq 1$ and the number of feedback bits $B \geq 0$.

  \subsection{Feedback Times}
  
  For a ``feedback resolution'' parameter $\theta >0$, we set the feedback times $t_1,t_2, \ldots t_T$ such that feedback occurs every $\theta n$ channel uses, i.e., $t_k = k \theta n$ for $k \in \{1,2, \ldots, T\}$ and $T = \frac{1}{\theta}-1$. In turn, we set the size of the feedback alphabet such that $\log_2 |\mathcal{Z}| = \lfloor B/T \rfloor = \lfloor \frac{\theta B}{1-\theta} \rfloor$. In summary, Bob sends $\lfloor \frac{\theta B}{1-\theta} \rfloor$ bits of feedback at each time $t = k \theta n \in \mathcal{T} \triangleq \{\theta n, 2\theta n, \ldots, (1-\theta)n\}$. We assume that $\theta n$ and $\frac{1}{\theta}-1$ are integer values.

  \subsection{Encoding Function}
  
  Alice uses her feedback to adapt her encoding every $\theta n$ channel uses at time $t = k \theta n \in \mathcal{T}$. Recall that for a message $m \in [q^{Rn}]$ and feedback symbols $(z_1, \ldots, z_{T}) \in \mathcal{Z}^{T}$, the codeword $\mathcal{C}(m;z_1,\ldots,z_{T})$ is a concatenation of $T+1 = 1/\theta$ chunks of \textit{sub-codewords},
  \begin{equation} \nonumber
  \begin{aligned}
  \mathcal{C}(m;z_1,z_2,\ldots,z_{T}) =  \mathcal{C}_1(m) \circ \mathcal{C}_2(m;z_1) \circ \cdots \circ \mathcal{C}_{T+1}(m;z_{T})
  \end{aligned}
  \end{equation}
  where $\mathcal{C}_k(m,z_{k-1})$ is the sub-codeword of $\theta n$ symbols corresponding to message $m$ and feedback $z_{k-1}$. Compared to the more general description of sub-codewords presented in our channel model of Section \ref{sec:cwf}, the sub-codeword $\mathcal{C}_k(m,z_{k-1})$ of our encoding function construction depends only on $m$ and the most recent feedback $z_{k-1}$, and does not depend on past feedback $(z_1, \ldots, z_{k-2})$.

  \subsection{Feedback Function}
  
   To describe our feedback function, we first make an observation about Bob's ability to observe erasures and errors in his received sequence at any given time. Let $\lambda_t$ denote the number of symbol erasures in Bob's received sequence $(y_1,y_2, \ldots, y_t)$ at time $t \in \mathcal{T}$. We remark that Bob can count the number of symbol erasures $\lambda_t$ in the received sequence $(y_1,y_2, \ldots, y_t)$, and thus, he can identify the number of unerased symbols $t-\lambda_t$ in $(y_1,\ldots,y_t)$. Crucially, however, Bob cannot count the number of \textit{symbol errors} $(t-\lambda_t)p_t$ in the unerased symbols in $(y_1,\ldots,y_t)$: here, $p_t$ is defined as the fraction of the $t-\lambda_t$ unerased symbols in the received sequence $(y_1,\ldots,y_t)$ that are in error. Following the naming convention established by \cite{Chen2014,Chen2019}, we refer to $p_t$ as the adversary's \textit{(error) trajectory} as it describes the path or strategy taken by the adversary over time to confuse Bob's decoder. 
   
   To proceed with our code construction, we define a \textit{reference trajectory} $\hat{p}_t$ as a place holder for the unobservable trajectory $p_t$. While $\hat{p}_t$ is in general not equal to $p_t$, we carefully choose this reference trajectory (details in Section \ref{sec:ref_traj}) such that for any value of $\{p_t\}_{t=1}^n$ chosen by the adversary, $\hat{p}_t \approx p_t$ for some $t \in \mathcal{T}$. Other properties of $\hat{p}_t$ are discussed below. We assume that $\hat{p}_t$ is known to Alice, Bob and the adversary. 
  
  At each time $t = k \theta n \in \mathcal{T}$, given the received word $(y_1,y_2, \ldots,y_{t}= y_{k\theta n})$, Bob computes the feedback symbol $z_{k} = f_{k}(y_1,\ldots,y_t)$ via the following process. First, Bob identifies the first time $t_{\mathrm{min}} \in \mathcal{T}$ such that the number of unerased symbols $t_{\mathrm{min}} - \lambda_{t_{\mathrm{min}}}$ exceeds $n(R+\epsilon_L)$ for some small parameter $\epsilon_L > 0$. If $t < t_{\mathrm{min}}$, then Bob sends the feedback $z_k =0$ to inform Alice that there are few erasures in his received sequence. If instead $t \geq t_{\mathrm{min}}$, Bob forms a list $\mathcal{L}_k$ of all messages $m' \in [q^{Rn}]$ that have a codeword prefix $\mathcal{C}_1(m') \circ \cdots \circ \mathcal{C}_k(m';z_{k-1})$ within a Hamming distance $(t-\lambda_t) \hat{p}_t$ of the received sequence prefix $\bm{y}_{\mathrm{prefix},k} \triangleq (y_1,\ldots,y_{t} = y_{k n \theta})$.\footnote{Recall that our definition of Hamming distance accounts for symbol erasures in the received word by only measuring the distance between $\mathcal{C}_1(m') \circ \cdots \circ \mathcal{C}_k(m';z_{k-1})$ and $\bm{y}_{\mathrm{prefix},k}$ over the positions where $\bm{y}_{\mathrm{prefix},k}$ is not erased.} In our analysis, we show that our choice of trajectory $\hat{p}_t$ ensures that $\mathcal{L}_k$ contains at most $L$ messages for some constant $L \geq 1$. Subsequently, for the integer $k_{\mathrm{min}}$ defined as $t_{\mathrm{min}} = k_{\mathrm{min}} \theta n$, Bob forms a super list
  \begin{equation}
  \mathcal{L}^{\mathrm{super}}_k = \bigcup_{i=k_{\mathrm{min}}}^k \mathcal{L}_i
  \end{equation}
  of all messages contained in his lists $\{ \mathcal{L}_i \}_{i=k_{\mathrm{min}}}^k$ up to sub-codeword $k$. It is clear that $\mathcal{L}^{\mathrm{super}}_k$ contains at most $(k-k_{\mathrm{min}})L$ messages which, in turn, is bounded above by a constant $\frac{1}{\theta} L$ independent of $n$. After forming the super list, Bob chooses feedback $z_{k} \in \mathcal{Z}$ such that for every unique pair $m',m'' \in \mathcal{L}^{\mathrm{super}}_k$, the $(k+1)^{\text{th}}$ sub-codewords corresponding to $m'$ and $m''$ are \textit{far apart} such that they satisfy the \textit{Feedback Distance Condition}:
  \begin{equation} \label{eq:fb_distance}
  \begin{aligned}
  d_H\left(\mathcal{C}_{k+1}(m';z_{k}), C_{k+1}(m'';z_{k}) \right) > \theta n \frac{(q-1)(1-\epsilon_L)(1+\epsilon_L)}{q}
  \end{aligned}
  \end{equation}
  where $\epsilon_L>0$ is a small ``list-decodability'' parameter that will be chosen in our analysis. If multiple such $z_k$ exist, Bob may choose any $z_k$ in this multiple. This choice of $z_{k}$ together with the choice of $\hat{p}_t$ helps Bob to identify and remove messages from $\mathcal{L}^{\mathrm{super}}_{T}$ that are not Alice's true message $m$. If the above $z_{k}$ does not exist, then Bob declares a decoding error. Otherwise, Bob sends $z_k$ to Alice via the feedback channel before Alice's next transmission at time $t+1$.
  
  \subsection{Decoding}
  
  %Bob waits to perform decoding until all $n$ symbols have been received. Decoding is performed as the following 3 step procedure. In step 1, $k$ is initialized to $k_{\mathrm{min}} = \frac{t_{\mathrm{min}}}{\theta n}$. In step 2, list-refinement is performed. Let $\bm{y}_{\mathrm{suffix},k} = (y_{k+1}, \ldots, y_n)$ denote the received word suffix at time $t = k\theta n$. If there exists a message $m' \in \mathcal{L}_k$ such that
  %\begin{equation}
  %\begin{aligned}
  %&d\left( \bm{y}_{\mathrm{suffix},k}, C_{k+1}(m';z_{k+1}) \circ \cdots C_{1/\theta}(m';z_{1/\theta})  \right) \\
  %& \hspace{2em} \leq \frac{(q-1)(1-\epsilon_L)(1+\epsilon_{\theta})}{2q} - \frac{nr - \lambda_t}{2},
  %\end{aligned}
  %\end{equation}
  %then the decoder outputs $\hat{m}=m'$ and decoding halts. In step 3, $k$ is incremented. If $k = n$, then a decoding error is declared. Otherwise, the procedure goes to step 2.
  
  Bob waits to perform decoding until all $n$ symbols have been received. At the time of decoding, Bob possesses a set of lists $\{\mathcal{L}_k\}_{k=k_{\mathrm{min}}}^T$ which serve as a short list of candidate messages for Alice's message $m$. In our analysis, we show that our choice of $\hat{p}_t$ ensures that $m$ is contained in $\mathcal{L}_k$ for some $k \in \{k_{\mathrm{min}},k_{\mathrm{min}}+1,\ldots,T \}$. Thus, Bob's decoding is a list-refinement (or list-disambiguation) process that uses information in the received sequence suffix $\bm{y}_{\mathrm{suffix},k} = (y_{k\theta n+1} = y_{t+1},y_{t+2}, \ldots,y_n)$ to infer which messages in list $\mathcal{L}_k$ are not $m$, and subsequently, remove these messages from $\mathcal{L}_k$.
  
  Decoding is as follows. Bob performs list-refinement on his set of lists $\{ \mathcal{L}_k \}$ by evaluating for each $k = k_{\mathrm{min}},\ldots, T$ the distance between codeword suffixes corresponding to messages in $\mathcal{L}_k$. Specifically, for the received sequence suffix $\bm{y}_{\mathrm{suffix},k} = (y_{t+1}, \ldots, y_n)$, Bob finds the \textit{decoding point} $t^* = k^* \theta n \in \mathcal{T}$ that is the smallest integer $t^* \in \{t_{\mathrm{min}},\ldots,(1-\theta) n\}$ such that there exists a message $m' \in \mathcal{L}_{k^*}$ satisfying the \textit{Decoding Distance Condition}:
  \begin{equation} \label{eq:dec_cond}
  \begin{aligned}
  d_H\left(C_{k^*+1}(m';z_{k^*}) \circ \cdots \circ C_{T+1}(m';z_{T}),\bm{y}_{\mathrm{suffix},k^*}  \right)  \leq \frac{(q-1)(1-\epsilon_L)(1+\epsilon_{L})(n-t^*)}{2q} - \frac{nr - \lambda_{t^*}}{2}.
  \end{aligned}
  \end{equation}
  A decoding error is declared if $t^*$ does not exist or if there are two or more $m' \in \mathcal{L}_{k^*}$ such that (\ref{eq:dec_cond}) holds. If no decoding error is declared, then there exists a unique $m' \in \mathcal{L}_{k^*}$ such that (\ref{eq:dec_cond}) holds, and in turn, the decoder outputs $\hat{m} = m'$.
  
  We remark that Bob's choice of feedback $z_{k}$ to satisfy the Feedback Distance Condition (c.f. (\ref{eq:fb_distance})) plays a critical role during our analysis for ensuring that the Decoding Distance Condition (\ref{eq:dec_cond}) only holds if $m' = m$ and does not hold if $m' \in \mathcal{L}_k \setminus \{ m \}$. 

  \subsection{Relationship to Codes for Causal Adversarial Channels}
  
  Lastly, we remark that the above encoding and decoding procedures are inspired by coding scheme constructions for causal adversarial channels (without feedback) proposed in \cite{Chen2014,Chen2019}. In particular, we borrow the idea of establishing a reference trajectory $\hat{p}_t$ and performing iterative list-decoding and list-refinement based on this reference trajectory. 
  
  In the causal adversarial model of \cite{Chen2014,Chen2019}, a feedback channel is not available for Bob to send the symbol $z_k$, and thus, Alice must instead choose a private random value for $z_k$ (unknown to Bob or the adversary) prior to encoding. In turn, the achievability analysis of \cite{Chen2014,Chen2019} permits decoding to fail with some small probability over the choice of $z_k$. Hence, the benefit of $O(1)$-bit feedback is that Alice and Bob can agree on a common value of $z_k$, allowing the above scheme to decode successfully with certainty.  %In particular, if $m' \in \mathcal{L}_k$ and $m' \neq m$, then
  %\begin{align}
  %& d\left(\mathcal{C}_{k+1}(m';z_{k+1}) \circ \cdots \circ \mathcal{C}_{1/\theta}(m';z_{1/\theta}), \bm{y}_{\mathrm{suffix},k} \right) \nonumber \\
  %& = \sum_{i = k+1}^{1/\theta} d \left(\mathcal{C}_{i}(m';z_i) \mathcal{C}_{i}(m;z_i) + \bm{e}_k \right) \nonumber \\
  %& > \sum_{i=k+1}^{1/\theta} \theta n \frac{(q-1)(1-\epsilon_L)(1+\epsilon_L)}{q} \\
  %& = 
  %\end{align}
  
  \section{Proof of Theorem \ref{thm:cap_bf}: Lower Bound} \label{sec:lb}

   The setup of the proof is as follows. In the sequel, we fix the alphabet size $q \geq 2$, the error fraction $p \in [0,\frac{q-1}{2q})$ and the erasure fraction $r \in [0,\frac{q-1}{q})$ such that $2p + r < \frac{q-1}{q}$. Let $\epsilon_R > 0$ be any rate backoff parameter and set the rate
  \begin{equation} \nonumber
  R = \min_{\bar{p}\in[0,p]} \left[ \alpha(\bar{p})\left(1-H_q\left(\frac{\bar{p}}{\alpha(\bar{p})} \right)\right) \right] - \epsilon_R.
  \end{equation}
  Let the blocklength $n \geq 1$. 

  \subsection{Overview of Achievability Proof}

  The proof uses the achievability coding scheme presented in Section \ref{sec:scheme}. The goal of the proof is to show that for any message $m \in [q^{Rn}]$ transmit by Alice and for any adversarial strategy $\adv \in \mathcal{ADV}$, Bob does not declare a decoding error and Bob's decoder outputs $\hat{m} = m$. Recall that Bob declares a decoding error if one of the three events occurs:
  \begin{enumerate}[label=\textit{E\arabic*)},itemindent=*]
  \item After receiving the $k^{\text{th}}$ sub-codeword for some $k \in [T]$, there does not exist a  feedback symbol $z_k \in \mathcal{Z}$ that satisfies the Feedback Distance Condition (\ref{eq:fb_distance}).
  \item The decoding point $t^* = k^* \theta n$ does not exist.
  \item The list $\mathcal{L}_{k^*}$ contains two or more messages that satisfy the Decoding Distance Condition (\ref{eq:dec_cond}).
  \end{enumerate}

  We remark that each of the above three events critically depends on the choice of \textit{reference trajectory} $\hat{p}_t$, which we recall, is used by Bob as a placeholder for the unobservable trajectory $p_t$ chosen by the adversary. The reference trajectory expression that we adopt is the following. Recall that $t_{\mathrm{min}}$ is defined as the smallest integer in $\mathcal{T} \triangleq \{\theta n, 2 \theta n, \ldots, (1-\theta) n\}$ such that $t_{\mathrm{min}} - \lambda_{t_{\mathrm{min}}} > n(R+\epsilon_L)$.

  \begin{definition}[Reference Trajectory] \label{def:rt}
  For $t \in \mathcal{T}$, define $\hat{p}_t =0$ if $t<t_{\mathrm{min}}$. Otherwise, if $t \geq t_{\mathrm{min}}$, define $\hat{p}_t$ to be the unique value in $[0,\frac{q-1}{q}]$ such that
  $(t-\lambda_t)(1-H_q(\hat{p}_t)) = n(R+\epsilon_L)$.
  \end{definition}

  We remark that the definition of $t_{\mathrm{min}}$ ensures that $\hat{p}_t$ is well-defined for all $t \geq t_{\mathrm{min}}$. Towards showing that the decoding error events E1 through E3 do not occur under the above choice of $\hat{p}_t$, we show that the following two conditions hold for all $t \in \mathcal{T}$ and $t \geq t_{\mathrm{min}}$: the \textit{List-Decoding Condition} 
  \begin{equation} \label{eq:LD}
  (t-\lambda_t)\left(1-H_q(\hat{p}_t) \right) - n \epsilon_L \geq n R
  \end{equation}
  and the \textit{List-Refinement Condition}
  \begin{equation} \label{eq:LR}
  nr - \lambda_t + 2(np - (t-\lambda_t) \hat{p}_t) \leq \frac{(q-1)(1-\epsilon_L)(n-t)}{q}.
  \end{equation} 

  The List-Decoding Condition (\ref{eq:LD}) is a condition on the size of Bob's list-decoding radius $(t-\lambda_t)\hat{p}_t$ and the coding rate $R$ which ensures that the number of messages in the list $\mathcal{L}_{k}$ is bounded above by some constant $L \geq 1$ for all appropriate $k$. In the analysis, the existence of constant $L$ ensures that a constant sized feedback alphabet $\mathcal{Z}$ (or equivalently, a constant sized number of bits $B$) is sufficient for the Feedback Distance Condition to hold (and thus, E1 does not occur). Note that we can immediate verify that $\hat{p}_t$ satisfies the List-Decoding Condition via inspection of Definition \ref{def:rt}.
  
  The List-Refinement Condition (\ref{eq:LR}) states that as long as $p_t \geq \hat{p}_t$ then the number of symbol errors and erasures in the received suffix $\bm{y}_{\mathrm{suffix},k}$ is bounded. In turn, $\bm{y}_{\mathrm{suffix},k}$ satisfies a number of nice distance properties: $\bm{y}_{\mathrm{suffix},k}$ is close to the codeword suffix corresponding to Alice's message $m$ and is far apart from other codeword suffixes corresponding to messages in Bob's list $\mathcal{L}_k$. Thus, when $p_t \geq \hat{p}_t$, Bob can use these distance properties to \textit{refine} or disambiguate the list $\mathcal{L}_k$. In the analysis, we show that if $m \in \mathcal{L}_k$ and $p_t \geq \hat{p}_t$ then List-Refinement Condition together with the List-Decoding Condition and the Feedback Distance Condition ensure that the codeword suffix corresponding to Alice's message $m$ is the unique message in $\mathcal{L}_k$ that is \textit{close} to $\bm{y}_{\mathrm{suffix},k}$, i.e., $m$ is the unique message that satisfies the Decoding Distance Condition. 
  
  A potential issue in this list-refinement process is the fact that, in general, $m$ is not in the list $\mathcal{L}_k$ when $p_t \geq \hat{p}_t$. Alice's message $m$ is only in $\mathcal{L}_k$ if Bob's list-decoding radius $(t-\lambda_t)\hat{p}_t$ is greater than the number of symbol errors $(t-\lambda_t)p_t$, i.e., if $p_t \leq \hat{p}_t$. Thus, for decoding to succeed, it is necessary that $p_t \approx \hat{p}_t$ for some time step $t \in \mathcal{T}$. In the analysis, we show that for any trajectory $p_t$ chosen by the adversary, $p_t \approx \hat{p}_t$ is guaranteed to occur for some time $t \in \mathcal{T}$ and $t \geq t_{\mathrm{min}}$. In fact, we show that the decoding point $t^*$ is the smallest value $t^* \in \mathcal{T}$ and $t^* \geq t_{\mathrm{min}}$ such that $p_{t^*} \approx \hat{p}_{t^*}$. Here, the role of the slack variables $\epsilon_L$ and $\theta$ is to allow the approximation of $p_{t^*}$ by $\hat{p}_{t^*}$ to be tuned arbitrarily well.

  Lastly, we briefly discuss the steps taken to show that $p_{t^*} \approx \hat{p}_{t^*}$. We break the analysis into two separate cases, each depending on the value of the trajectory at time $t_{\mathrm{min}}$. The first cases occurs when $p_{t_{\mathrm{min}}} \leq \hat{p}_{t_{\mathrm{min}}}$, in which case we say the trajectory is a \textit{low-type trajectory}. We show that when a low-type trajectory occurs, $p_{t_{\mathrm{min}}} \approx \hat{p}_{t_{\mathrm{min}}}$ and $t^* = t_{\mathrm{min}}$. The second case occurs when $p_{t_{\mathrm{min}}} > \hat{p}_{t_{\mathrm{min}}}$, in which case we say the trajectory is a \textit{high-type trajectory}. We show that when a high-type trajectory occurs, $p_t$ is eventually smaller than $\hat{p}_t$ and, in turn, $t^*$ coincides with the smallest time $t^{**} \in \mathcal{T}$ and $t^{**} \geq t_{\mathrm{min}}$ such that $p_{t^{**}} \leq \hat{p}_{t^{**}}$.

  \subsection{Existence of a Good Code} \label{sec:codebooks}
  
  In this section, we establish the existence of an encoding function $\mathcal{C}_1 \circ \cdots \circ \mathcal{C}_{T+1}$ that allows Alice and Bob to communicate with zero error. We establish this existence by using a \textit{random coding argument} in which codewords of $\mathcal{C}_1 \circ \cdots \circ \mathcal{C}_{T+1}$ are picked randomly, and in turn, certain desirable properties of the encoding function are shown to hold with high probability over the choice of encoding function. We focus on properties related to the list-decodability and codeword distance.
  
  In this section, we let encoding function $\mathcal{C}_1 \circ \cdots \circ \mathcal{C}_{T+1}$ be randomly drawn from the following distribution. For every message $m \in [q^{Rn}]$ and every feedback tuple $(z_1,\ldots,z_{T}) \in \mathcal{Z}^T$, the codeword $\mathcal{C}_1(m) \circ \mathcal{C}_2(m;z_1) \circ \cdots \circ \mathcal{C}_{T+1}(m;z_{T})$ is uniformly distributed in $\mathcal{Q}^n$. Furthermore, codewords are independent for all $(m,z_1,\ldots,z_{T}) \in [q^{Rn}] \times \mathcal{Z}^T$.

  The following lemma will prove useful.  
  \begin{lemma}[\hspace{-0.01em}{\cite[Theorem~12.1.3]{Cover1991}}] \label{thm:bin_ub}
  For $0\leq p \leq \frac{q-1}{q}$,
  \begin{equation} \nonumber
  \sum_{i=0}^{np} {n \choose i}(q-1)^i \leq q^{n H_q(p)}. 
  \end{equation}
  \end{lemma}

  We first state a formal definition of list-decodability.
  
  \begin{definition}[List-Decodability]
  For $k \in [T+1]$, $\lambda_t \in [0,t]$ and $L \geq 1$, an encoding function $\mathcal{C}_1 \circ \cdots \circ \mathcal{C}_k$ is list-decodable for up to $(t-\lambda_t)\hat{p}_t$ symbol errors and $\lambda_t$ symbol erasures with list size $L$ if for every received sequence prefix $\bm{y}_{\mathrm{prefix},k} = (y_1, \ldots,y_t)$ with at most $\lambda_t$ symbol erasures and the corresponding feedback tuple $(z_1,z_2,\ldots,z_{k-1},\cdot) = f(\bm{y}_{\mathrm{prefix},k},\cdot)$ there are at most $L$ messages $m \in [q^{Rn}]$ such that the codeword suffix $\mathcal{C}_1(m) \circ \mathcal{C}_2(m;z_1) \cdots \circ \mathcal{C}_{k}(m;z_{k-1})$ is within a Hamming distance $(t-\lambda_t)\hat{p}_t$ of $\bm{y}_{\mathrm{prefix},k}$.
  \end{definition}
  
  The proof of the following lemma relies on the reference trajectory satisfying the List-Decoding Condition (\ref{eq:LD}).
  
  \begin{lemma}[\hspace{-0.01em}{\cite[Claim~III.14]{Chen2019}}] \label{thm:LD}
  Let $t = k \theta n \in \mathcal{T}$ and let $\lambda_t \in \{0,\ldots,rn\}$ such that $t \geq t_{\mathrm{min}}$. With probability at least $1 - q^{-n}$ over the design of encoding function, the encoding function $\mathcal{C}_1 \circ \cdots \circ \mathcal{C}_{k}$ is list-decodable for up to $(t-\lambda_t)\hat{p}_t$ symbol errors and $\lambda_t$ symbol erasures with list size $L = O(1/\epsilon_L)$. Proof is in Appendix \ref{sec:LD_proof}.
  \end{lemma}
  
  \begin{lemma}[List-decodable Code] \label{thm:LD2}
  With probability at least $1 - \frac{1-\theta}{\theta} n q^{-n}$ over the design of encoding function, for every $t = k \theta n \in \mathcal{T}$ and every $\lambda_t \in [0,rn]$ such that $t \geq t_{\mathrm{min}}$, the encoding function $\mathcal{C}_1 \circ \cdots \circ \mathcal{C}_{k}$ is list-decodable for up to $(t-\lambda_t)\hat{p}_t$ symbol errors and $\lambda_t$ symbol erasures with list size
  $L = O(1/\epsilon_L)$.
  \end{lemma}
  
  \begin{proof}[Proof of Lemma \ref{thm:LD2}] By Lemma \ref{thm:LD}, for fixed $t= k \theta n \in \mathcal{T}$ and $\lambda_t \in [0,t]$ such that $t \geq t_{\mathrm{min}}$, the probability that $\mathcal{C}_1 \circ \cdots \circ \mathcal{C}_k$ is not list decodable with the stated parameters is at most $2^{-n}$. By a union bound over all $t = k \theta n \in \mathcal{T}$ and all $\lambda_t \in \{0,1, \ldots, rn\}$ such that $t \geq t_{\mathrm{min}}$, the probability that $\mathcal{C}_1 \circ \cdots \circ \mathcal{C}_k$ is not list decodable for some pair $(t,\lambda_t)$ is at most $|\mathcal{T}|rn2^{-n} < \frac{1-\theta}{\theta} n 2^{-n}$.
  \end{proof}

  \begin{lemma} \label{thm:cb_fb}
  Suppose that the size of the feedback symbol set $\mathcal{Z}$ is large enough such that
  \begin{equation} \label{thm:cb_fb_Z}
  |\mathcal{Z}| > \frac{RL}{\theta^2\left(1-H_q\left(\frac{(q-1)(1-\epsilon_L)(1+\epsilon_L)}{q} \right)\right)}.
  \end{equation}
  Then with probability at most $q^{-\Omega(n)}$ over the design of encoding function, for any $k \in [T]$ and for any $\mathcal{L}_k^{\mathrm{super}} \subset [q^{Rn}]$ such that $|\mathcal{L}_k^{\mathrm{super}}| \leq L/\theta$, there exists feedback symbol $z_{k} \in \mathcal{Z}$ such that the sub-codewords corresponding to feedback $z_{k}$ and messages in $\mathcal{L}^{\mathrm{super}}_{k}$ satisfy the Feedback Distance Condition: for all unique pairs $m',m'' \in \mathcal{L}^{\mathrm{super}}_{k}$
  \begin{equation} \label{eq:cb_fb_dist}
  \begin{aligned}
  d_H\left(\mathcal{C}_{k+1}(m';z_{k}),\mathcal{C}_{k+1}(m'';z_{k})\right) > \theta n \frac{(q-1)(1-\epsilon_L)(1+\epsilon_L)}{q}.
  \end{aligned}
  \end{equation}
  \end{lemma}
  
  \begin{proof}[Proof of Lemma \ref{thm:cb_fb}]
  %The proof relies on the reference trajectory satisfying the List-Refinement Condition (\ref{eq:LR}) \hl{what??}. 
  For $k \in [T]$, unique messages $m',m'' \in [q^{Rn}]$ and feedback symbol $z_{k} \in \mathcal{Z}$, define $\mathcal{D}_{k}(m',m'',z_{k})$ to be the event that (\ref{eq:cb_fb_dist}) does not hold. Let $\mathbb{P}$ denote the probability measure w.r.t. the design of the random encoding function $\mathcal{C}_1 \circ \cdots \circ \mathcal{C}_{T+1}$. Fix the size of the feedback set $\mathcal{Z}$ to be a constant such that (\ref{thm:cb_fb_Z}) holds. To prove Lemma \ref{thm:cb_fb}, we show that the probability
  \begin{equation} 
  \mathbb{P} \left( \bigcup_{k=1}^{T} \bigcup_{\substack{\mathcal{L}_{k}^{\mathrm{super}}\subset [q^{Rn}] \\: |\mathcal{L}_{k}^{\mathrm{super}}| = L/\theta}} \bigcap_{z_{k} \in \mathcal{Z}} \bigcup_{\substack{m',m'' \in \mathcal{L}_{k}^{\mathrm{super}} \\ :m' \neq m''}} D_{k}(m',m'',z_{k}) \right) \nonumber
  \end{equation}
  tends to $0$ exponentially fast as the blocklength $n$ tends to infinity.
  
  By a simple union bound, the above probability is bounded above by
  \begin{equation} \nonumber
  \sum_{k=1}^T \sum_{\substack{\mathcal{L}_k^{\mathrm{super}}\subset [q^{Rn}] \\ : |\mathcal{L}_k^{\mathrm{super}}| = L/\theta}}\mathbb{P} \left( \bigcap_{z_{k} \in \mathcal{Z}} \bigcup_{\substack{m',m'' \in \mathcal{L}_{k}^{\mathrm{super}} \\ :m' \neq m''}} D_{k}(m',m'',z_{k}) \right)
  \end{equation}
  which, in turn, following the independence of sub-codewords over feedback $z_{k} \in \mathcal{Z}$, is equal to 
  \begin{equation} \label{eq:cb_fb_2}
  \sum_{k=1}^T \sum_{\substack{\mathcal{L}_k^{\mathrm{super}}\subset [q^{Rn}] \\ : |\mathcal{L}_k^{\mathrm{super}}| = L/\theta}}\mathbb{P} \left( \bigcup_{\substack{m',m'' \in \mathcal{L}_{k}^{\mathrm{super}}\\ :m' \neq m''}} \hspace{-1em} D_{k}(m',m'',z_{k}) \right)^{|\mathcal{Z}|}.
  \end{equation}
  An additional union bound yields that (\ref{eq:cb_fb_2}) is bounded above by
  \begin{align} \label{eq:cb_fb_3}
  &\sum_{k=1}^T \sum_{\substack{\mathcal{L}_k^{\mathrm{super}}\subset [q^{Rn}] \\ : |\mathcal{L}_k^{\mathrm{super}}| = L/\theta}} \left[ \sum_{\substack{m',m'' \in \mathcal{L}_{k}^{\mathrm{super}}\\ :m' \neq m''}} \mathbb{P} \left( D_{k}(m',m'',z_{k}) \right) \right]^{|\mathcal{Z}|}
  \end{align}
  In (\ref{eq:cb_fb_3}), we remark that the quantity $\mathbb{P}\left(D_{k}(m',m'',z_{k}) \right)$ follows the binomial cumulative distribution function (CDF), i.e.,
  \begin{align}
  \mathbb{P} \left( D_{k}(m',m'',z_{k}) \right) = \hspace{-3em} \sum_{i=0}^{\theta n \frac{(q-1)(1-\epsilon_L)(1+\epsilon_L)}{q}} \hspace{-1em} {\theta n \choose i} \left(1 - \frac{1}{q}\right)^i \left(\frac{1}{q}\right)^{\theta n-i}\nonumber 
  \end{align}
  which in turn, following Lemma \ref{thm:bin_ub}, is bounded above by
  $$q^{\theta n \left(1- H_q\left( \frac{(q-1)(1-\epsilon_L)(1+\epsilon_L)}{q} \right) \right)}.$$ Thus, following a simple counting exercise,   (\ref{eq:cb_fb_3}) is bounded above by
  \begin{align}
  &T{q^{Rn} \choose L/\theta} {L/\theta \choose 2}^{|\mathcal{Z}|} q^{-|\mathcal{Z}|\theta n \left(1-H_q \left( \frac{(q-1)(1-\epsilon_L)(1+\epsilon_L)}{q} \right) \right)} \nonumber \\
  & \leq T {L/\theta \choose 2}^{|\mathcal{Z}|} q^{\frac{RnL}{\theta}-|\mathcal{Z}|\theta n \left(1-H_q \left( \frac{(q-1)(1-\epsilon_L)(1+\epsilon_L)}{q} \right) \right)} \label{eq:cb_fb_5}
  \end{align}
  where (\ref{eq:cb_fb_5}) follows from the binomial coefficient inequality ${q^{Rn} \choose L/\theta} \leq q^{RnL/\theta}$. In turn, following the assumption that $|\mathcal{Z}|$ is lower bounded by (\ref{thm:cb_fb_Z}), (\ref{eq:cb_fb_5}) is bounded above by $q^{-\Omega(n)}$.
  \end{proof}

  \subsection{Properties of Reference Trajectory $\hat{p}_t$} \label{sec:ref_traj}
 
  In this section, we characterize the following properties of the the reference trajectory $\hat{p}_t$. First, we show that $\hat{p}_t$ satisfies the List-Refinement Condition (\ref{eq:LR}) for sufficiently list-decoding parameter $\epsilon_L>0$ (c.f. Lemma \ref{thm:LR_suff}). Thus, $\hat{p}_t$ satisfies both the List-Decoding and the List-Refinement Conditions. Next, we show that reference trajectory $\hat{p}_t$ is eventually greater than the trajectory $p_t$ for any $p_t$ chosen by the adversary (c.f. Lemma \ref{thm:traj_large}). In turn, we show that the choice of reference trajectory guarantees that the number of symbol erasures $\lambda_t$ and the number of symbol errors $(t-\lambda_t) p_t$ in the received suffix is bounded for all time $t \in \mathcal{T}$ and $t \geq t_{\mathrm{min}}$ below some time threshold $t^{**}$ (c.f., Lemma \ref{thm:rt_summary}). Later, in Section \ref{sec:combine}, we will show that Bob's decoding point $t^*$ coincides with this threshold $t^{**}$.

  Recall that the reference trajectory $\hat{p}_t$ is given by Definition \ref{def:rt}. For the reader's convenience, we restate the definition here. 

  \begin{definition}[Reference Trajectory (Definition \ref{def:rt} Repeated)]
  For $t \in \mathcal{T}$, define $\hat{p}_t =0$ if $t<t_{\mathrm{min}}$. Otherwise, if $t \geq t_{\mathrm{min}}$, define $\hat{p}_t$ to be the unique value in $[0,\frac{q-1}{q}]$ such that $(t-\lambda_t)(1-H_q(\hat{p}_t)) = n(R+\epsilon_L)$.
  \end{definition}

  We remark that reference trajectory $\hat{p}_t$ is increasing in $t$ for all $t \in \mathcal{T}$ and $t \geq t_{\mathrm{min}}$. This follows from the fact that $H_q(x) \in [0,1]$ is increasing in $x \in [0,\frac{q-1}{q}]$.

  \begin{lemma}[List-Refinement Condition] \label{thm:LR_suff}
  For small enough $\epsilon_L>0$, $\hat{p}_t$ satisfies the List-Refinement Condition (\ref{eq:LR}) for all $t \in \mathcal{T}$ such that $t \geq t_{\mathrm{min}}$.
  \end{lemma}
  
  \begin{proof}[Proof of Lemma \ref{thm:LR_suff}]
  %Let $\{\lambda_t\}_{t=1}^n$ be an increasing sequence in $\{0,1, \ldots, rn\}$ where $\lambda_1 \leq 1$ and $\lambda_{t+1} - \lambda_{t} \leq 1$ for all $t \in [n-1]$. 
  We begin by restating the List-Refinement Condition (\ref{eq:LR}). By a rearrangement of terms, (\ref{eq:LR}) can be shown to be equivalent to
  \begin{equation} \label{eq:LR2}
   \hat{p}_t \geq \frac{nr-\lambda_t}{2(t-\lambda_t)} + \frac{np}{t-\lambda_t} - \frac{(q-1)(1-\epsilon_L)(n-t)}{2q(t-\lambda_t)} .
  \end{equation}
  We remark that if the right-hand-side (RHS) of (\ref{eq:LR2}) is non-positive for small enough $\epsilon_L>0$ and all $t \in \mathcal{T}$ and $t \geq t_{\mathrm{min}}$ then (\ref{eq:LR2}) trivially holds following the fact that $\hat{p}_t \geq 0$. Thus, we only need to consider the case where the RHS of (\ref{eq:LR2}) is positive for all $\epsilon_L>0$ and some bad pairs $(t,\lambda_t) \in \mathcal{B}$ where $ \mathcal{B} \subseteq \mathcal{T} \times \{0,1,\ldots,rn\}$ and $t \geq t_{\mathrm{min}}$. In the sequel, suppose that the RHS of (\ref{eq:LR2}) is positive for all $\epsilon_L>0$ and all $(t,\lambda_t) \in \mathcal{B}$.
  
  In turn, we remark that the left-hand-side (LHS) of (\ref{eq:LR2}) is bounded between $(0,\frac{q-1}{q}]$ following the fact that $\hat{p}_t \leq \frac{q-1}{q}$. In turn, we can utilize the fact that $H_q(\cdot)$ is an increasing function on $[0,\frac{q-1}{q}]$ and rewrite (\ref{eq:LR2})
  \begin{equation} \label{eq:LR3}
  \begin{aligned}
  & H_q\left(\hat{p}_t \right) \geq \\
  & \hspace{1em} H_q \left( \frac{nr-\lambda_t}{2(t-\lambda_t)} + \frac{np}{t-\lambda_t} - \frac{(q-1)(1-\epsilon_L)(n-t)}{2q(t-\lambda_t)} \right).
  \end{aligned}
  \end{equation}
  Following the definition of $\hat{p}_t$, (\ref{eq:LR3}) is equivalent to the following inequality:
  \begin{align}
  &n(R+\epsilon_L) \nonumber \\
  & \triangleq (t- \lambda_t)(1 - H_q(\hat{p}_t)) \nonumber \\
  & \leq (t - \lambda_t) \left(1 - H_q\left(\beta_t - \frac{(q-1)(1-\epsilon_L)(n-t)}{2q(t-\lambda_t)} \right) \right) \label{eq:f_ineq}
  \end{align}
  where $\beta_t = \frac{nr-\lambda_t}{2(t-\lambda_t)} + \frac{np}{t-\lambda_t}$.
  
  In the sequel, we show that the above inequality holds for small enough $\epsilon_L >0$. Recall that $R = \min_{\bar{p} \in [0,p]} [\alpha(\bar{p})(1-H_q(\frac{\bar{p}}{\alpha(\bar{p})}))] - \epsilon_R$ where $\alpha(\bar{p}) \triangleq 1 -\frac{2q}{q-1}(p-\bar{p}) - \frac{q}{q-1}r$. Thus, $R$ is bounded above by $\alpha(\bar{p})(1-H(\frac{\bar{p}}{\alpha(\bar{p})})) - \epsilon_R$ for any $\bar{p} \in [0,p]$. Setting $\bar{p}$ such that $\alpha(\bar{p}) = \frac{t-\lambda_t}{n}$, it follows that $\frac{\bar{p}}{\alpha(\bar{p})} = \beta_t - \frac{(q-1)(n-t)}{2q(t-\lambda_t)} + \frac{\lambda_t}{2q(t-\lambda_t)}$ and, in turn,
  \begin{align}
  & n(R+\epsilon_L) \nonumber \\
  & \leq (t-\lambda_t) \left(1 - H_q\left(\beta_{t} - \frac{(q-1)(n-t)}{2q(t-\lambda_t)}+ \frac{\lambda_t}{2q(t - \lambda_t)}\right)\right) + n (\epsilon_L - \epsilon_R) \nonumber \\
  & \stackrel{(a)}{\leq} (t-\lambda_t) \left(1 - H_q\left(\beta_{t} - \frac{(q-1)(n-t)}{2q(t-\lambda_t)}\right)\right)  + n(\epsilon_L-\epsilon_R) \nonumber
  \end{align}
  where (a) follows from the fact that $H_q(\cdot)$ is an increasing function on $[0,\frac{q-1}{2q}]$ and the fact that $\beta_t - \frac{(q-1)(n-t)}{2q(t-\lambda_t)} + \frac{\lambda_t}{2q(t-\lambda_t)} = \frac{\bar{p}}{\alpha(\bar{p})} \leq \frac{q-1}{q}$.\footnote{The second fact can be seen via the following argument. Define the function $g(\gamma) \triangleq \frac{\gamma}{\alpha(\gamma)}$ for $\gamma \in [0,p]$ which has the derivative $\frac{\partial g(\gamma)}{\partial \gamma}$ equal to $\frac{(q-1)(-1+q(1-2p-r))}{(2pq-q-2 \gamma q+qr+1)^2}$. Following that $2p+r < \frac{q-1}{q}$, we have that $\frac{\partial g(\gamma)}{\partial \gamma} > 0$ for all $\gamma \in [0,p]$, and thus, $\max\limits_{\gamma \in [0,p]} g(\gamma) = g(p) = \frac{p}{1 - \frac{q}{q-1}r} < \frac{\frac{q-1}{2q}-\frac{r}{2}}{1 - \frac{q}{q-1}r}=\frac{q-1}{2q}$.} Define
  \begin{align}
  &\eta(t,\lambda_t,\epsilon_L) \triangleq H_q\left(\beta_t - \frac{(q-1)(1-\epsilon_L)(n-t)}{2q(t-\lambda_t)} \right) - H_q\left( \beta_t - \frac{(q-1)(n-t)}{2q(t-\lambda_t)} \right) \nonumber
  \end{align}
  and note that the quantity $\eta(t,\lambda_t,\epsilon_L)$ tends to $0$ as $\epsilon_L$ tends to $0$ following the fact that the term $\frac{\epsilon_L(q-1)(n-t)}{(2q(t-\lambda_t))} \leq \frac{\epsilon_L(q-1)(n-t)}{(2q(t_{\mathrm{min}}-\lambda_{t_{\mathrm{min}}}))} < \frac{\epsilon_L(q-1)}{R+\epsilon_L}$ tends to 0 and the fact that $H_q(\cdot)$ is a continuous function. With this notation in hand, we rewrite (a):
  \begin{equation} \label{eq:LD5}
  \begin{aligned}
  &(t - \lambda_t) \left(1 - H_q\left(\beta_t - \frac{(q-1)(1-\epsilon_L)(n-t)}{2q(t-\lambda_t)} \right) \right) + (t-\lambda_t) \eta(t,\lambda_t,\epsilon_L) + n(\epsilon_L - \epsilon_R).
  \end{aligned}
  \end{equation}
  We remark that for small enough $\epsilon_L>0$, the term $(t-\lambda_t)\eta(t,\lambda_t,\epsilon_L) + n(\epsilon_L - \epsilon_R)$ is negative for all bad pairs $(t,\lambda_t) \in \mathcal{B}$. Thus, for small enough $\epsilon_L>0$, (\ref{eq:LD5}) is bounded above by (\ref{eq:f_ineq}) for all $(t,\lambda_t) \in \mathcal{B}$, and thus, the List-Refinement Condition (\ref{eq:LR3}) holds for all $(t,\lambda_t) \in \mathcal{B}$.
  \end{proof}

  The following lemma states that the reference trajectory $\hat{p}_t$ is eventually greater than the adversary's trajectory $p_t$.
  
  \begin{lemma} \label{thm:traj_large}
  Define $t_{\mathrm{max}}$ to be the largest integer in $\mathcal{T}$ such that $t_{\mathrm{max}} \leq n(1- \epsilon_L)$. For small enough $\epsilon_L>0$ and for any $\theta \in (0,\epsilon_L)$, $(t_{\mathrm{max}} - \lambda_{t_{\mathrm{max}}}) \hat{p}_{t_{\mathrm{max}}} \geq np$.
  \end{lemma}
  
  \begin{proof}[Proof of Lemma \ref{thm:traj_large}]
  We remark that the definition of $t_{\mathrm{max}}$ along with the definition of $\mathcal{T} \in \{\theta n, 2 \theta n, \ldots, (1-\theta)n \}$ implies that $t_{\mathrm{max}}$ is bounded below by $n(1-\epsilon_L -\theta)$.
  
  Note that the inequality $(t_{\mathrm{max}}-\lambda_{t_{\mathrm{max}}}) \hat{p}_{t_{\mathrm{max}}} \geq np$ is equivalent to
  \begin{equation} \label{eq:tr_lg_2}
  H_q(\hat{p}_{t_{\mathrm{max}}}) \geq H_q\left(\frac{pn}{t_{\mathrm{max}} - \lambda_{t_{\mathrm{max}}}}\right)
  \end{equation}
  following the fact that $H_q(\cdot)$ is increasing on $[0,\frac{q-1}{q}]$ and $\hat{p}_{t_{\mathrm{max}}} \leq \frac{q-1}{q}$. In turn, following the definition of $\hat{p}_t$, (\ref{eq:tr_lg_2}) is equivalent to
  \begin{equation} \label{eq:tr_lg_3}
  \begin{aligned}
  & n(R+\epsilon_L) \\
  & \triangleq (t_{\mathrm{max}} - \lambda_{t_{\mathrm{max}}})(1 - H_q(\hat{p}_{t_{\mathrm{max}}})) \\
  & \leq (t_{\mathrm{max}} - \lambda_{t_{\mathrm{max}}})\left(1 - H_q\left( \frac{pn}{t_{\mathrm{max}}-\lambda_{t_{\mathrm{max}}}} \right)\right).
  \end{aligned}
  \end{equation}

  To prove Lemma \ref{thm:traj_large}, we show that for small enough $\epsilon_L>0$ and any $\theta \in (0,\epsilon_L)$, the inequality (\ref{eq:tr_lg_3}) holds. Recall that $R = \min_{\bar{p} \in [0,p]} [\alpha(\bar{p})(1-H_q(\frac{\bar{p}}{\alpha(\bar{p})}))] - \epsilon_R$ and, thus, $R$ is bounded above by $\alpha(\bar{p})(1-H_q(\frac{\bar{p}}{\alpha(\bar{p})})) - \epsilon_R$ for any $\bar{p} \in [0,p]$. Setting $\bar{p}$ such that $\alpha(\bar{p}) = \frac{t_{\mathrm{max}}-\lambda_{t_{\mathrm{max}}}}{n}$, it follows that $\frac{\bar{p}}{\alpha(\bar{p})} = \beta + \frac{np}{t_{\mathrm{max}}-\lambda_{t_{\mathrm{max}}}} - \frac{(q-1)(n-t)}{2q(t_{\mathrm{max}}-\lambda_{t_{\mathrm{max}}})}$ where $\beta = \frac{nr-\lambda_{t_{\mathrm{max}}}}{2(t_{\mathrm{max}}-\lambda_{t_{\mathrm{max}}})} + \frac{\lambda_{t_{\mathrm{max}}}}{2q(t_{\mathrm{max}}-\lambda_{t_{\mathrm{max}}})}$ and, in turn, for small enough $\epsilon_L>0$ any $\theta \leq \epsilon_L$,
  \begin{align}
  n(R+\epsilon_L) \leq (t_{\mathrm{max}} - \lambda_{t_{\mathrm{max}}})\left(1-H_q\left(\beta + \frac{np}{t_{\mathrm{max}}-\lambda_{t_{\mathrm{max}}}} - \frac{(q-1)(n-t_{\mathrm{max}})}{2q(t_{\mathrm{max}}-\lambda_{t_{\mathrm{max}}})} \right)\right) + n(\epsilon_L-\epsilon_R). \nonumber
  \end{align}
  Note that since $\beta \geq 0$, the quantity inside the function $H_q(\cdot)$ is bounded below by $\frac{np}{t_{\mathrm{max}}-\lambda_{t_{\mathrm{max}}}} - \frac{(q-1)(n-t_{\mathrm{max}})}{2q(t_{\mathrm{max}}-\lambda_{t_{\mathrm{max}}})}$. We remark that this lower bound is non-negative for small enough $\epsilon_L$ and $\theta \leq \epsilon_L$ since $t_{\mathrm{max}} \geq n(1-\epsilon_L -\theta)$, and thus, $n-t_{\mathrm{max}} \leq n(\epsilon_L+\theta)$. In turn, for small enough $\epsilon_L>0$ and any $\theta \in (0,\epsilon_L)$
  \begin{align}
  n(R+\epsilon_L) \leq (t_{\mathrm{max}} - \lambda_{t_{\mathrm{max}}})\left(1-H_q\left( \frac{np}{t_{\mathrm{max}}-\lambda_{t_{\mathrm{max}}}} - \frac{(q-1)(n-t_{\mathrm{max}})}{2q(t_{\mathrm{max}}-\lambda_{t_{\mathrm{max}}})} \right)\right) + n(\epsilon_L-\epsilon_R). \label{eq:tr_lg_4}
  \end{align}
  Define 
   \begin{equation} \nonumber
   \begin{aligned}
   \eta(\epsilon_L) =  H_q\left( \frac{np}{t_{\mathrm{max}} - \lambda_{t_{\mathrm{max}}}} \right) - H_q\left(\frac{np}{t_{\mathrm{max}}-\lambda_{t_{\mathrm{max}}}} -  \frac{(q-1)(n-t_{\mathrm{max}})}{2q(t_{\mathrm{max}}-\lambda_{t_{\mathrm{max}}})} \right)
   \end{aligned}
   \end{equation}
   and note that $\eta(\epsilon_L)$ tends to $0$ as $\epsilon_L$ tends to 0; this follows from the fact that $t_{\mathrm{max}} \geq n(1 - \epsilon_L-\theta)$ and, thus, $\frac{n-t_{\mathrm{max}}}{t_{\mathrm{max}} - \lambda_{t_{\mathrm{max}}}} \leq \frac{n(\epsilon_L+\theta)}{t_{\mathrm{min}} - \lambda_{t_{\mathrm{min}}}} < \frac{\epsilon_L+\theta}{R+\epsilon_L}$ tends to $0$ as $\epsilon_L$ and $\theta \leq \epsilon_L$ tend to $0$. Using this notation, the RHS of (\ref{eq:tr_lg_4}) can be written as
   \begin{equation}
   \begin{aligned}
   (t_{\mathrm{max}}-\lambda_{t_{\mathrm{max}}})\left(1-H_q\left(\frac{pn}{t_{\mathrm{max}}-\lambda_{t_{\mathrm{max}}}}\right)\right) + (t_{\mathrm{max}} - \lambda_{t_{\mathrm{max}}})\eta(\epsilon_L) + n(\epsilon_L-\epsilon_R)
   \end{aligned}
   \end{equation}
   where the term $(t_{\mathrm{max}} - \lambda_{t_{\mathrm{max}}})\eta(\epsilon_L) + n(\epsilon_L-\epsilon_R)$ is negative for small enough $\epsilon_L$ and any $\theta \in (0,\epsilon_L)$. In conclusion, for such values of $\epsilon_L, \theta$, the desired inequality (\ref{eq:tr_lg_3}) holds.
  \end{proof}

  We now show that the number of symbol erasures $\lambda_t$ and symbol errors $(t-\lambda_t)p_t$ in the received suffix is bounded for all time $t \in \mathcal{T}$ and $t \geq t_{\mathrm{min}}$ below some time threshold $t^*$. To do this, analyze the reference trajectory $\hat{p}_t$ under two mutually exclusive cases: when $p_t$ is a \textit{low-type trajectory} (i.e., $p_{t_{\mathrm{min}}} \leq \hat{p}_{t_{\mathrm{min}}}$ and when $p_t$ is a \textit{high-type trajectory} (i.e., $p_{t_{\mathrm{min}}} > \hat{p}_{t_{\mathrm{min}}}$).
  
  \begin{lemma}[Low-Type Trajectory] \label{thm:low_traj}
  Let $\epsilon_L>0$ be small enough and let $\theta \in (0,\epsilon_L)$. If $\hat{p}_{t_{\mathrm{min}}} \geq p_{t_{\mathrm{min}}}$, then
  \begin{equation} \label{eq:low_traj}
  nr - \lambda_{t_{\mathrm{min}}} + 2(np - (t_{\mathrm{min}}-\lambda_{t_{\mathrm{min}}})p_{t_{\mathrm{min}}}) \leq \frac{(q-1)(1-\epsilon_L)(n-t_{\mathrm{min}})}{q}.
  \end{equation}
  \end{lemma}
  
  \begin{proof}[Proof of Lemma \ref{thm:low_traj}] Suppose that $\hat{p}_{t_{\mathrm{min}}} \geq p_{t_{\mathrm{min}}}$. To show that (\ref{eq:low_traj}) holds for small enough $\epsilon_L>0$ and any $\theta \in (0,\epsilon_L)$, it is sufficient to show that the following inequality holds:
  \begin{equation} \label{eq:low_traj_2}
  t_{\mathrm{min}} - \lambda_{t_{\mathrm{min}}} \leq n\left(1 - \frac{q}{(q-1)(1-\epsilon_L)}r - \frac{2q}{(q-1)(1-\epsilon_L)}p\right)
  \end{equation}
  which follows from (\ref{eq:low_traj}) and the trivial bounds $p_{t_{\mathrm{min}}} \geq 0$ and $\frac{q}{(q-1)(1-\epsilon_L)} \geq 1$. Note that the RHS of (\ref{eq:low_traj_2}) is positive for small enough $\epsilon_L>0$ following the bound $r+2p<\frac{q-1}{q}$, and in turn, let $\epsilon_L$ be sufficiently small for such a condition. 
  
  We proceed by showing that (\ref{eq:low_traj_2}) holds. Recall that the definition of $t_{\mathrm{min}}$ implies that $(t_{\mathrm{min}} - n\theta) - \lambda_{(t_{\mathrm{min}} - n\theta)} \leq n(R+\epsilon_L)$, and in turn, $t_{\mathrm{min}} - \lambda_{t_{\mathrm{min}}} \leq n(R+\epsilon_L+2\theta)$. In turn, following the definition of $R$, we have that 
  \begin{align}
  & t_{\mathrm{min}} - \lambda_{t_{\mathrm{min}}} \nonumber \\
  & \leq n(R+\epsilon_L+2\theta) \nonumber \\
  & = n\left(\min_{\bar{p}\in[0,p]} \left[\alpha(\bar{p})\left(1-H_q\left(\frac{\bar{p}}{\alpha(\bar{p})}\right)\right)\right] - \epsilon_R + \epsilon_L + 2\theta \right) \nonumber \\
  &\leq n\left(1-\frac{q}{q-1}r - \frac{2q}{q-1}p - \epsilon_R + \epsilon_L+2\theta \right) \label{eq:low_traj_3}
  \end{align}
  where (\ref{eq:low_traj_3}) follows by setting the minimization parameter $\bar{p}$ to $0$. Furthermore, for small enough $\epsilon_L>0$ and for any $\theta \in (0,\epsilon_L)$, (\ref{eq:low_traj_3}) is bounded above by the RHS of (\ref{eq:low_traj_2}). In conclusion, for these choices of $\epsilon_L,\theta$, the inequality (\ref{eq:low_traj_2}) holds.
  \end{proof}
  
  \begin{lemma}[High-Type Trajectory] \label{thm:high_traj}
  Let $\epsilon_L>0$ be small enough and let $\theta \in (0,\epsilon_L)$ be small enough. If $\hat{p}_{t_{\mathrm{min}}} < p_{t_{\mathrm{min}}}$, then for the smallest integer $t^{**} \in \mathcal{T}$ and $t^{**} \geq t_{\mathrm{min}}$ such that $\hat{p}_{t^{**}} \geq p_{t^{**}}$,
  \begin{equation} \label{thm:high_traj_ineq}
  \begin{aligned}
  nr - \lambda_{t^{**}} + 2(np - (t^{**}-\lambda_{t^{**}})p_{t^{**}}) \leq \frac{(q-1)(1-\epsilon_L)(1+\epsilon_{L})(n-t^{**})}{q}.
  \end{aligned}
  \end{equation}
  \end{lemma}
  
  \begin{proof}[Proof of Lemma \ref{thm:high_traj}] Suppose that $\hat{p}_{t_{\mathrm{min}}}<p_{t_{\mathrm{min}}}$. We first note that $t^{**}$ exists for small enough $\epsilon_L>0$ and any $\theta \in (0,\epsilon_L)$ following Lemma \ref{thm:traj_large}. In our proof of Lemma \ref{thm:high_traj}, we take $\epsilon_L>0$ to be small enough such that Lemma \ref{thm:traj_large} holds, and we take $t^{**} \in \mathcal{T}$ to be the smallest integer greater than $t_{\mathrm{min}}$ such that $\hat{p}_{t^{**}} \geq p_{t^{**}}$.
  
  We first show that $\hat{p}_{t} \approx \hat{p}_{t-\theta n}$ for all $t \in \mathcal{T}$ and $t>t_{\mathrm{min}}$. Let $t \in \mathcal{T}$ such that $t > t_{\mathrm{min}}$ and let $t' = t-\theta n$. Recall that by the definition of $\hat{p}_t$,
  $\hat{p}_t = H_q^{-1}\left(1 - \frac{n(R+\epsilon_L)}{t-\lambda_t} \right)$
  where $H_q^{-1}$ is the \textit{inverse q-ary entropy function}, i.e., for $y \in [0,1]$, $H_q^{-1}(y) = x$ where $H_q(x) = y$ for $x \in [0,\frac{q-1}{q}]$. We remark that $H_q^{-1}$ is continuous on the compact space $[0,1]$ and, thus, $H_q^{-1}$ is \textit{uniformly continuous} on $[0,1]$, i.e., for any $\delta_1>0$, $\exists \delta_2>0$ such that $|H_q^{-1}(y_1) - H_q^{-1}(y_2)| < \delta_1$ for any $y_1,y_2 \in [0,1]$ such that $|y_1-y_2|<\delta_2$. Thus, for any $\delta_1>0$, $\exists \delta_2 >0$ such that
  \begin{align}
  |\hat{p}_t - \hat{p}_{t'}| \triangleq \left|H_q^{-1}\left( 1- \frac{n(R+\epsilon_L)}{t-\lambda_t} \right) -H_q^{-1}\left( 1- \frac{n(R+\epsilon_L)}{t'-\lambda_{t'}} \right) \right| < \delta_1 \nonumber
  \end{align}
  if 
  \begin{equation} \label{eq:un_cont}
  \left|\frac{n(R+\epsilon_L)}{t-\lambda_t} - \frac{n(R+\epsilon_L)}{t'-\lambda_{t'}}\right| < \delta_2.
  \end{equation}
  By a rearrangement of terms, the LHS of (\ref{eq:un_cont}) is equal to $n(R+\epsilon_L)|\frac{t-\lambda_t - t'+\lambda_{t'}}{(t-\lambda_t)(t'-\lambda_{t'})}|$ which, in turn, is bounded above by $\frac{\theta}{R+\epsilon_L}$ following the fact that $|t - \lambda_t - t' + \lambda_{t'}| \leq \theta n$ and $t-\lambda_t \geq t' - \lambda_{t'} \geq t_{\mathrm{min}} -\lambda_{t_{\mathrm{min}}} > n(R+\epsilon_L)$. Hence, for however small $\delta_2$ is, we can always choose $\theta$ small enough such that inequality (\ref{eq:un_cont}) holds for all $t \in \mathcal{T}$ and $t>t_{\mathrm{min}}$. Hence, for a value $\delta_1>0$ which we will soon set, we let $\theta \in (0,\epsilon_L)$ be small enough such that $\hat{p}_t - \delta_1 < \hat{p}_{t- \theta n}$ for all $t \in \mathcal{T}$ and $t > t_{\mathrm{min}}$. We will refer to the inequality $\hat{p}_t - \delta_1 < \hat{p}_{t-\theta n}$ for all $t \in \mathcal{T}$ and $t \geq t_{\mathrm{min}}$ as Inequality (a).
  
  We are now ready to prove Lemma \ref{thm:high_traj}. By the above choice of $t^{**}$, we have that $\hat{p}_{t^{**}-\theta n} < p_{t^{**}-\theta n}$. We will refer to this inequality as Inequality (b). Furthermore, since the number of symbol errors is increasing over time, we have that $(t^{**}-\lambda_{t^{**}})p_{t^{**}} \geq (t^{**}-\theta n - \lambda_{t^{**}-\theta n})p_{t^{**} - \theta n}$ and, in turn, since $\lambda_{t^{**}-\theta n} \leq \lambda_{t^{**}}$, we have that $p_{t^{**} - \theta n} \leq (1+ \frac{\theta n}{t^{**}-\lambda_{t^{**}}-\theta n})p_{t^{**}}$. Given that $p_{t^{**}} < 1$ and the $t_{\mathrm{min}}$ definitional inequalities $t^{**}-\lambda_{t^{**}} > t_{\mathrm{min}} - \lambda_{t_{\mathrm{min}}} > n(R+\epsilon_L)$, it follows that $p_{t^{**}-\theta n} < p_{t^{**}} + \frac{\theta}{R+\epsilon_L-\theta}$ which we refer to as Inequality (c). By putting together Inequality (a) through (c), we have
  \begin{equation} \nonumber
  \hat{p}_{t^{**}} - \delta_1 \stackrel{(a)}{<} \hat{p}_{t^{**}-\theta n} \stackrel{(b)}{<} p_{t^{**}-\theta n} \stackrel{(c)}{\leq} p_{t^{**}}+\frac{\theta}{R+\epsilon_L-\theta} \nonumber,
  \end{equation}
  and, thus, $\hat{p}_{t^{**}} < p_{t^{**}} + \delta_1 + \frac{\theta}{R+\epsilon_L-\theta}$. In turn, we have the following inequalities which resemble the desired inequality (\ref{thm:high_traj_ineq}): for small enough $\epsilon_L>0$,
  \begin{align}
  & nr-\lambda_{t^{**}} + 2\left(np-(t^{**}-\lambda_{t^{**}})\left(p_{t^{**}} + \delta_1 + \frac{\theta}{R+\epsilon_L-\theta}\right)\right) \nonumber \\
  & \leq nr - \lambda_{t^{**}} + 2(np - (t^{**}-\lambda_{t^{**}})\hat{p}_{t^{**}}) \nonumber \\
  & \leq \frac{(q-1)(1-\epsilon_L)(n-t^{**})}{q} \label{eq:high_traj_lr}
  \end{align}
  where (\ref{eq:high_traj_lr}) follows for small enough $\epsilon_L>0$ from the claim that $\hat{p}_{t^{**}}$ satisfies the List-Refinement Condition (c.f. Lemma \ref{thm:LR_suff}). Thus, for small enough $\epsilon_L>0$ such that (\ref{eq:high_traj_lr}) holds, and for small enough $\delta_1 > 0$ and small enough $\theta \in (0,\epsilon_L)$ such that the inequality 
  \begin{align}
  2(t^{**}-\lambda_{t^{**}}) \left(\delta_1 + \frac{\theta}{R+\epsilon_L-\theta} \right) \leq \frac{(q-1)(1-\epsilon_L)\epsilon_L(n-t^{**})}{q} \nonumber
  \end{align}
  holds (where we use the fact that $n-t^{**} \geq n-t_{\mathrm{max}} \geq \epsilon_L n$), we have the desired inequality
  \begin{align}
  nr - \lambda_{t^{**}} + 2(np - (t^{**}-\lambda_{t^{**}})p_{t^{**}}) \leq \frac{(q-1)(1-\epsilon_L)(1+\epsilon_L)(n-t^{**})}{q}. \nonumber
  \end{align}
  \end{proof}
  
  \begin{lemma} \label{thm:traj_lr}
  Let $\epsilon_L>0$ be small enough. For any $t \in \mathcal{T}$ such that $t \geq t_{\mathrm{min}}$, if $\hat{p}_t < p_t$ then
  \begin{equation} \nonumber
  nr-\lambda_t + 2(np - (t-\lambda_t)p_t) \leq \frac{(q-1)(1-\epsilon_L)(n-t)}{q}.
  \end{equation}
  \end{lemma}
  
  Lemma \ref{thm:traj_lr} immediately follows from Lemma \ref{thm:LR_suff}, i.e., the claim that $\hat{p}_t$ satisfies the List-Refinement Condition (\ref{eq:LR}) for small enough $\epsilon_L>0$. The next lemma ties together each of the above reference trajectory lemmas.
  
  \begin{lemma}[Reference Trajectory Summary] \label{thm:rt_summary}
  Let $\epsilon_L>0$ be small enough and let $\theta \in (0,\epsilon_L)$ be small enough. Let $t^{**} \in \mathcal{T}$ be the smallest integer $t^{**} \geq t_{\mathrm{min}}$ such that $\hat{p}_{t^{**}} \geq p_{t^{**}}$. Then $t^{**}$ exists and for every $t \in \mathcal{T}$ such that $t \in [t_{\mathrm{min}},t^{**}]$,
  \begin{equation} \label{eq:rt_summary}
  \begin{aligned}
  nr - \lambda_{t} + 2(np - (t-\lambda_{t})p_{t})  \leq \frac{(q-1)(1-\epsilon_L)(1+\epsilon_{L})(n-t)}{q}.
  \end{aligned}
  \end{equation}
  \end{lemma}
  
  \begin{proof}[Proof of Lemma \ref{thm:rt_summary}]
  The existence of $t^{**}$ for small enough $\epsilon_L > 0$ and any $\theta \in (0,\epsilon_L)$ follows directly from Lemma \ref{thm:traj_large}. We show that inequality (\ref{eq:rt_summary}) holds for every $t \in \mathcal{T}$ such that $t \in [t_{\mathrm{min}},t^{**}]$ by considering two cases. 
  
  (Case 1): Suppose that the trajectory $p_t$ is a low-type trajectory, i.e., $\hat{p}_{t_{\mathrm{min}}} \geq p_{t_{\mathrm{min}}}$. Then $t^{**} = t_{\mathrm{min}}$ and for $t = t^{**} = t_{\mathrm{min}}$, the inequality (\ref{eq:rt_summary}) holds  for small enough $\epsilon_L > 0$ and any $\theta \in (0,\epsilon_L)$ following Lemma \ref{thm:low_traj}. This completes Case 1. 
  
  (Case 2): Suppose instead that the trajectory $p_t$ is a high-type trajectory, i.e., $\hat{p}_{t_{\mathrm{min}}}< p_{t_{\mathrm{min}}}$. Then the definition of $t^{**}$ implies that $\hat{p}_t < p_t$ for $t \in [t_{\mathrm{min}},t^{**}]$, and in turn, (\ref{eq:rt_summary}) holds for small enough $\epsilon_L>0$ following Lemma \ref{thm:traj_lr}. For $t = t^{**}$, (\ref{eq:rt_summary}) holds for small enough $\epsilon_L>0$ and small enough $\theta \in (0,\epsilon_L)$ following Lemma \ref{thm:high_traj}.
  \end{proof}
  
  \subsection{Combining Intermediate Results} \label{sec:combine}
  
  We now prove the lower bound of Theorem \ref{thm:cap_bf} by combining the above intermediate results. In the sequel, we adopt the following setup:
  \begin{enumerate}
  \item Recall that the parameters $p \in [0,\frac{q-1}{2q})$ and $r \in [0,\frac{q-1}{q})$ are fixed such that $2p + r \leq \frac{q-1}{q}$. For the rate backoff parameter $\epsilon_R > 0$, the rate is set such that
  \begin{equation} \nonumber
  R = \min_{\bar{p}\in[0,p]} \left[ \alpha(\bar{p})\left(1-H_q\left(\frac{\bar{p}}{\alpha(\bar{p})} \right)\right) \right] - \epsilon_R.
  \end{equation}
  \item Choose the ``list-decoding'' parameter $\epsilon_L>0$ small enough to satisfy the hypotheses of Lemma \ref{thm:LR_suff} and Lemma \ref{thm:rt_summary}.
  \item Set the ``feedback resolution'' parameter $\theta \in (0,\epsilon_L)$ to be small enough to satisfy the hypothesis of Lemma \ref{thm:rt_summary}. Define $T \triangleq \frac{1}{\theta}-1$.
  \item Set list-size $L$ to be the $O(1/\epsilon_L)$ parameter identified in Lemma \ref{thm:LD} and Lemma \ref{thm:LD2}.
  \item Set the size of the feedback set to be any constant such that
  \begin{equation} \nonumber
  |\mathcal{Z}| > \frac{RL}{\theta^2\left(1-H_q\left(\frac{(q-1)(1-\epsilon_L)(1+\epsilon_L)}{q} \right)\right)}
  \end{equation}
  where we remark that, in turn, the number of feedback bits $B$ is a constant such that $\lfloor B/T \rfloor = \lfloor \frac{\theta B}{1-\theta} \rfloor \triangleq \log_2 |\mathcal{Z}|$. Recall that the above bound on $|\mathcal{Z}|$ is chosen to satisfy the hypothesis of Lemma \ref{thm:cb_fb}.
  \item For every blocklength $n$ large enough such the following encoding function exists, let the encoding function $\mathcal{C}_1 \circ \cdots \circ \mathcal{C}_{T+1}$ be any (deterministic) function with the following 2 properties:
  \begin{enumerate}
  \item \textit{List-decodable Code}: For every $t = k \theta n \in \mathcal{T}$ and every $\lambda_t \in [0,rn]$ such that $t \geq t_{\mathrm{min}}$, the encoding function prefix $\mathcal{C}_1 \circ \cdots \mathcal{C}_k$ is list-decodable for up to $(t-\lambda_t)\hat{p}_t$ symbol errors and $\lambda_t$ symbol erasures with list size $L$.
  \item \textit{Feedback Distance Condition}: For every $k \in [T]$, there exists a feedback symbol $z_{k} \in \mathcal{Z}$ such that for every unique message pair $m',m'' \in \mathcal{L}^{\mathrm{super}}_{k}$,
  \begin{equation} \label{eq:fb_distance2}
  \begin{aligned}
  d_H\left(\mathcal{C}_{k+1}(m';z_{k}),\mathcal{C}_{k+1}(m'';z_{k})\right)  > \theta n \frac{(q-1)(1-\epsilon_L)(1+\epsilon_L)}{q}. 
  \end{aligned}
  \end{equation}
  \end{enumerate}
  We remark that following Lemma \ref{thm:LD2} and Lemma \ref{thm:cb_fb}, an encoding function picked at random has both of the above 2 properties with probability at least $1 - q^{-\Omega(n)} - q^{-\Omega(n)} > 0$. In turn, this implies the existence of an encoding function with the above 2 properties.
  \end{enumerate}

  With the above setup, we prove the lower bound of Theorem \ref{thm:cap_bf} by showing that for any strategy chosen by the adversary, Bob does not declare a decoding error (i.e., the decoding process does not halt prior outputting $\hat{m}$) and Bob's decoder outputs Alice's message $m$.
  
  Recall that Bob declares a decoding error if at least one of the following two error events occur: \textbf{a)} the decoding point $t^*$ does not exist or \textbf{b)} there exists two or more $m' \in \mathcal{L}_{k^*}$ that satisfy the Decoding Distance Condition (\ref{eq:dec_cond}).\footnote{Recall that Bob may also declare a decoding error if for some $k \in [T]$, there does not exist a feedback symbol $z_{k} \in \mathcal{Z}$ such that every unique pair $m',m'' \in \mathcal{L}_k^{\mathrm{super}}$ satisfies the Feedback Distance Condition (\ref{eq:fb_distance2}). However, we have already ruled out this error event due to our choice of encoding function.} The following two lemmas state that these decoding error events will not occur for any adversarial strategy. 
  
  Recall that $t^{**} = k^{**} \theta n \in \mathcal{T}$ is the smallest integer $t^{**} \geq t_{0}$ such that $\hat{p}_{t^{**}} \geq p_{t^{**}}$. Furthermore, recall that $t^{**}$ exists (c.f. Lemma \ref{thm:rt_summary}).
  \begin{lemma} \label{thm:dec_exist}
  The decoding point $t^{*}$ is equal to $t^{**}$. Thus, $t^*$ exists.
  \end{lemma}
  
  \begin{proof}[Proof of Lemma \ref{thm:dec_exist}]
  Recall that $t^* = k^* \theta n \in \mathcal{T}$ is the smallest integer $t^* \geq t_{\mathrm{min}}$ such that there exists a message $m' \in \mathcal{L}_{k^*}$ that satisfies the \textit{Decoding Distance Condition} (c.f. (\ref{eq:dec_cond})): 
  \begin{equation} \label{eq:dec_exist_22}
  \begin{aligned}
  d_H\left(C_{k^*+1}(m';z_{k^*}) \circ \cdots C_{T+1}(m';z_{T}),\bm{y}_{\mathrm{suffix},k^*}  \right) \leq \frac{(q-1)(1-\epsilon_L)(1+\epsilon_{L})(n-t^*)}{2q} - \frac{nr - \lambda_{t^*}}{2}.
  \end{aligned}
  \end{equation}

  We first show that $t^* \leq t^{**}$. From the definition of $t^{**}$, Bob's list $\mathcal{L}_{k^{**}}$ must contain Alice's message $m$ since $\mathcal{L}_{k^{**}}$ is formed using a list-decoding radius $(t^{**} - \lambda_{t^{**}})\hat{p}_{t^{**}}$ which is greater than the number of symbol errors $(t^{**} - \lambda_{t^{**}})p_{t^{**}}$. Furthermore, by decomposing Bob's received sequence suffix $\bm{y}_{\mathrm{suffix},k^{**}}$ into Alice's codeword suffix $\bm{x}_{\mathrm{suffix},k^{**}}$ and the error-erasure suffix $\bm{e}_{\mathrm{suffix},k^{**}}$ of $\lambda_{t^{**}}$ symbol erasures and $(t^{**} - \lambda_{t^{**}}) p_{t^{**}}$ symbol errors, we have that
  \begin{align}
  &d_H\left(C_{k^{**}+1}(m;z_{k^{**}}) \circ \cdots C_{T+1}(m;z_{T}),\bm{y}_{\mathrm{suffix},k^{**}}  \right) \nonumber \\
  & = d_H(0,\bm{e}_{\mathrm{suffix},k^{**}}) \nonumber \\
  & \stackrel{(a)}{\leq} np - (t^{**} - \lambda_{t^{**}})p_{t^{**}} \nonumber \\
  & \stackrel{(b)}{\leq} \frac{(q-1)(1-\epsilon_L)(1+\epsilon_{L})(n-t^{**})}{2q} - \frac{nr - \lambda_{t^{**}}}{2} \label{eq:dec_exist_3} 
  \end{align}
  where (a) follows from the assumption that the adversary can induce at most $pn$ symbol errors, and (b) follows from Lemma \ref{thm:rt_summary}. In conclusion, (\ref{eq:dec_exist_22}) holds for $t^* = t^{**}$, and thus, $t^{**}$ is an upper bound of $t^{*}$.

  We now show that $t^* \geq t^{**}$. Suppose by contradiction that $t^* < t^{**}$. Then $\hat{p}_{t^*} < p_{t^*}$, and in turn, Alice's message $m$ is not in $\mathcal{L}_{k^*}$. In turn, by the definition of $t^*$, there exists a message $m'$ in $\mathcal{L}_{k^*}$ which is not equal to $m$ and which satisfies the Decoding Distance Condition (\ref{eq:dec_exist_22}). Using a trivial lower bound on Hamming distance, we note that the LHS of (\ref{eq:dec_exist_22}) is bounded such that
  \begin{align}
  d_H\left(\mathcal{C}_{k^*+1}(m';z_{k^*}) \circ \cdots \circ \mathcal{C}_{T+1}(m';z_{T}), \bm{y}_{\mathrm{suffix},k^*} \right)  \stackrel{(c)}{\geq} D - E \nonumber
  \end{align}
  where 
  \begin{align}
  D &\triangleq d_H\bigg(\mathcal{C}_{k^*+1}(m';z_{k^*}) \circ \cdots \circ \mathcal{C}_{T+1}(m';z_{T}), \mathcal{C}_{k^*+1}(m;z_{k^*}) \circ \cdots \circ \mathcal{C}_{T+1}(m;z_{T}) \bigg) \nonumber \\
  &= \sum_{i=k^*}^{T} d_H\left( \mathcal{C}_{i+1}(m';z_{k^*}), \mathcal{C}_{i+1}(m;z_{k^*})\right) \nonumber
  \end{align}
  is the Hamming distance between the codeword suffixes corresponding to messages $m'$ and $m$, and where $E$ is the total number of symbol erasures plus the total number of symbol errors inject by the adversary in the received suffix $\bm{y}_{\mathrm{suffix},k^*}$. Following the assumption that the adversary can inject up to $pn$ symbol errors and $rn$ symbol erasures, $E$ is bounded above by $(np - (t^*-\lambda_{t^*})p_{t^*}) + (nr-\lambda_{t^*})$. In turn, the LHS of (\ref{eq:dec_exist_22}) is bounded below by
  \begin{align}
  & D - \left(np - (t^* - \lambda_{t^*})p_{t^*} \right) - (nr - \lambda_{t^*}) \nonumber \\
  & \stackrel{(d)}{\geq} D -\frac{(q-1)(1-\epsilon_L)(1+\epsilon_L)(n-t^*)}{2q} - \frac{nr-\lambda_{t^*}}{2} \nonumber \\
  & \stackrel{(e)}{>} \frac{(q-1)(1-\epsilon_L)(1+\epsilon_L)(n-t^*)}{2q} - \frac{nr-\lambda_{t^*}}{2} \nonumber
  \end{align}
  where (d) follows from Lemma \ref{thm:rt_summary} and (e) follows from the Feedback Distance Condition (\ref{eq:fb_distance2}). However, (e) contradicts the supposition that $m'$ satisfies the Decoding Distance Condition (\ref{eq:dec_exist_22}). Thus, it must be true that $t* = t^{**}$.
  \end{proof}

  \begin{lemma} \label{thm:dec_success}
  Alice's message $m$ is contained in $\mathcal{L}_{k^*}$. Furthermore, $m$ is the unique message in $\mathcal{L}_{k^*}$ that satisfies the Decoding Distance Condition (\ref{eq:dec_cond}). Hence, Bob's decoding procedure outputs $\hat{m} = m$.
  \end{lemma}

  \begin{proof}[Proof of Lemma \ref{thm:dec_success}]
  By Claim \ref{thm:dec_exist}, the decoding point $t^* = k^* \theta n$ exists and is equal to $t^{**}$. In turn, following the definition of $t^{**}$, we have that $\hat{p}_{t^*} \geq p_{t^*}$. Recall that Bob forms the list $\mathcal{L}_{k^*}$ with a list-decoding radius $(t^{*} -\lambda_{t^*})\hat{p}_{t^*}$. In turn, the list-decoding radius is at least as large as the number of symbol errors $(t^*-\lambda_{t^*})p_{t^*}$, which implies that $m$ must be in $\mathcal{L}_{k^*}$.

  The claim that $m$ is the unique message in $\mathcal{L}_{k^*}$ that satisfies the Decoding Distance Condition follows from the Feedback Distance Condition (\ref{eq:fb_distance2}). To see this, suppose that $\mathcal{L}_{k^*}$ contains 2 or more messages and let $m' \in \mathcal{L}_{k^*} \setminus \{m\}$. Applying the same argument outlined in the proof of Lemma \ref{thm:dec_exist} (specifically, inequalities (c) through (e)), it follows that
  \begin{align}
  d_H\left(\mathcal{C}_{k^*+1}(m';z_{k^*}) \circ \cdots \circ \mathcal{C}_{T+1}(m';z_{T}), \bm{y}_{\mathrm{suffix},k^*} \right) > \frac{(q-1)(1-\epsilon_L)(1+\epsilon_L)(n-t^*)}{2q} - \frac{nr-\lambda_{t^*}}{2},\nonumber
  \end{align}
  and thus, $m'$ does not satisfy the Decoding Distance Condition.
  \end{proof}

  As a corollary of Lemma \ref{thm:dec_exist} and Lemma \ref{thm:dec_success}, the rate $R = C^{O(1)}_q(p,r) - \epsilon_R$ is (zero error) achievable. Since we have allowed $\epsilon_R>0$ to be arbitrarily small, this completes the proof of the lower bound of Theorem \ref{thm:cap_bf}.

  \section{Proof of Theorem \ref{thm:cap_bf}: Upper Bound} \label{sec:ub}

  We first present a summary of our proof followed by a detailed proof. In the sequel, we adopt the following setup. Let $q \geq 2$. Let $p \in [0,\frac{q-1}{2q}]$ and $r \in [0,\frac{q-1}{q}]$ be the fraction of symbol errors and erasures, respectively, where $2p + r \leq \frac{q-1}{q}$. The proof can be easily extended to account for the case where $2p+r > \frac{q-1}{q}$. For $\bar{p} \in [0,p]$, let $\alpha(\bar{p}) = 1 - \frac{2q}{q-1} (p-\bar{p}) - \frac{q}{q-1} r$. Define $$C_0 = \min_{\bar{p} \in [0,p]} \left[\alpha(\bar{p})\left(1-H_q\left(\frac{\bar{p}}{\alpha(\bar{p})}\right)\right)\right].$$ 
  
  \subsection{Overview of Converse Proof}
    
  Roughly, the aim of our proof is to show that for any $\epsilon_R >0$, rate $R = C_0  + \epsilon_R$ and any communication scheme using $O(1)$-bit feedback, there exists an adversarial strategy $\adv \in \mathcal{ADV}$ and a message $m \in \mathcal{M}$ such that Bob incorrectly decodes $m$ when sent by Alice. To be more precise with our aim, we define the notation of a \textit{confusable} message pair.
  
  For an $(n,Rn,B)$-code $\Psi = (\mathcal{C}_k,\phi,f_k,\mathcal{T},\mathcal{Z})$, the pair of (unique) messages $m_1,m_2 \in \mathcal{M}$ is said to be \textit{confusable} if there exists an adversarial channel mapping $\adv^* \in \mathcal{ADV}$ such that the codewords corresponding to both messages $m_1$ and $m_2$ are mapped to the same received sequence $\bm{y}^* \in \mathcal{Y}^n$, i.e., $$\bm{y}^* = \adv^*(\mathcal{C}(m_1;z_1^*,\ldots,z_T^*)) = \adv^*(\mathcal{C}(m_2;z^*_1,\ldots,z^*_T))$$ for the feedback sequence  $(z^*_1,\ldots,z^*_T) = f(\bm{y}^*)$. Thus, if some pair $(m_1,m_2)$ is confusable, for any decoder $\phi$ used by Bob, the adversary can induce a decoding error if Alice sends $m \in \{m_1,m_2\}$. Hence, for any $\epsilon_R > 0$ sufficiently small, rate $R = C_0  + \epsilon_R$, any integer constant $B \geq 1$ and for any $(n,Rn,B)$-code $\Psi = (\mathcal{C}_k,\phi,f_k,\mathcal{T},\mathcal{Z})$, we show that for large enough $n$ there exists a message pair $(m_1,m_2)$ that is confusable. To show the existence of this pair, we provide a construction of $\adv^*$ via the ``babble-and-push'' attack.

  Roughly, the ``babble-and-push'' attack is a two-stage attack executed sequentially in time. In the first stage or ``babble'' attack, the adversary ignores its information about message $m$ and code $\Psi$ by ``babbling'' or \textit{randomly} injecting a certain number of symbol errors into Alice's codeword. As a result of this first stage attack, Bob has some uncertainty of Alice's message $m$. In the second stage or ``push'' attack, the adversary chooses a target message $m' \neq m$ from Bob's uncertainty set and uses its remaining budget of symbol errors and its entire budget of symbol erasures to ``push'' Alice's codeword to the codeword corresponding to $m'$. Following the second stage, the goal of the ``babble-and-push'' analysis is to show that Bob cannot reliably decide whether Alice sent $m$ or $m'$. The following is a detailed summary of the ``babble-and-push'' attack.

  \subsection{Summary of ``Babble-and-Push'' Attack}

  \textbf{``Babble'' Attack}: Let $m$ denote Alice's transmitted message which is drawn uniformly from $\mathcal{M}$ and known to the adversary and let 
  \begin{equation} \label{eq:p_bar_ass}
  \bar{p} = \arg \min_{\bar{p} \in [0,p]} \left[\alpha(\bar{p}) \left(1-H_q\left( \frac{\bar{p}}{\alpha(\bar{p})} \right)\right) \right].
  \end{equation}
  For the first $b = (\alpha(\bar{p}) + \epsilon_R/2)n$ channel uses, the adversary ``babbles'' by randomly injecting $\bar{p}n$ symbol errors into Alice's codeword. More specifically, the adversary randomly chooses an index subset $\mathcal{S} \subset \{1,\ldots,b\}$ of size $\bar{p}n$, and subsequently chooses Bob's $i^{\text{th}}$ received symbol $y_{i}$ uniformly from $\mathcal{Q} \setminus \{ x_{i} \}$ for all $i \in \mathcal{S}$ and sets $y_{i} = x_{i}$ for all $i \in [b] \setminus \mathcal{S}$. 
  
  Let $\bm{x}_{\mathrm{b}} = (x_1, \ldots, x_b)$ and $\bm{y}_{\mathrm{b}} = (y_1,\ldots, y_b)$ denote Alice's transmitted codeword and Bob's received sequence up to time $b$, respectively, following the adversary's ``babble'' attack. Furthermore, let $T_{\mathrm{b}} \in [0,T]$ denote the number of rounds of feedback that occur up to time $b$ and let $z_1, \ldots, z_{T_{\mathrm{b}}}$ denote the feedback symbols sent during these $T_{\mathrm{b}}$ rounds.    
  
  \textbf{Additional Assumption}: In our analysis and without loss of generality, we consider a stronger communication scheme than initially described in the channel model. We strengthen the scheme by incrementing Bob's feedback budget by an additional bit and requiring Bob to send the extra bit of feedback immediately after channel use $b$. Thus, Alice and Bob use an $(n,Rn,B+1)$-code, where we ensure that the $T_{\mathrm{b}}^{\text{th}}$ feedback symbol $z_{T_{\mathrm{b}}}$ is sent immediately after the $b^{\text{th}}$ channel use (i.e., at the end of the ``babble'' attack). The purpose of this assumption is to ensure that the adversary's ``push'' attack does not begin in the middle of a sub-codeword, thus simplifying our analysis of the ``push'' attack. Note that this assumption does not affect the generality of the converse result, as Alice can always ``ignore'' the extra feedback.
   
   \textbf{``Push'' Setup}: For a message $m'$ and feedback sequence $z'_1,\ldots,z'_{T}$, we partition the codeword $\mathcal{C}(m';z_1,\ldots,z_T)$ into its first $b$ symbols and its last $n-b$ symbols. We denote the first $b$ symbols as 
  \begin{align}
   \mathcal{C}^{(\mathrm{b})}(m';z_1,\ldots,z_{T_{\mathrm{b}}-1}) \triangleq \mathcal{C}_1(m') \circ \mathcal{C}_2(m';z_1) \circ \cdots \circ \mathcal{C}_{T_{\mathrm{b}}}(m';z_1,\ldots,z_{T_{\mathrm{b}}-1}) \nonumber
   \end{align}
   and the last $n-b$ symbols as
   \begin{align}
   \mathcal{C}^{(\mathrm{p})}(m';z_1,\ldots,z_{T}) \triangleq \mathcal{C}_{T_{\mathrm{b}}+1}(m';z_1,\ldots,z_{T_{\mathrm{b}}}) \circ \cdots \circ \mathcal{C}_{T+1}(m';z_1,\ldots,z_{T}). \nonumber
   \end{align}
   
   Following the ``babble'' attack, the adversary constructs a set of message pairs that are confusable given the observed output sequence $\bm{y}_{\mathrm{b}}$. To do this, the adversary first constructs a set of all messages $m'$ such that the first $b$ symbols of codeword $\mathcal{C}(m';z_1,\ldots,z_{T_b},\cdot)$ are close to $\bm{y}_{\mathrm{b}}$. That is, the adversary constructs the set $$\mathcal{B}_{\bm{y}_{\mathrm{b}}} = \left\{ m' \in \mathcal{M}: d_H\left(\bm{y}_{\mathrm{b}}, \mathcal{C}^{(\mathrm{b})}(m';z_1,\ldots,z_{T_{\mathrm{b}}-1}) \right) = \bar{p}n \right \}.$$ 
  
  Next, for each sub-codeword index $k \in \{T_{\mathrm{b}}+1, \ldots, T+1\}$ which has yet to be sent and each feedback sequence $\bm{z}_{k-1}'$ that agrees with the feedback sent in the first $b$ channel uses, i.e., $$\bm{z}_{k-1}' \in \mathcal{Z}_{k-1}' \triangleq \{z_1\} \times \cdots \times \{z_{T_{\mathrm{b}}}\} \times \mathcal{Z}^{k-1-T_{\mathrm{b}}},$$ the adversary constructs a set of all message pairs $(m',m'') \in \mathcal{B}_{\bm{y}_{\mathrm{b}}}^2$ such that the $k^{\text{th}}$ sub-codewords corresponding to $m$, $\bm{z}_{k-1}'$ and $m'$, $\bm{z}_{k-1}'$ are close. Specifically, the adversary constructs the set
  \begin{align} \nonumber
  \mathcal{D}_{k,\bm{z}_{k-1}'} = \{ (m',m'') \in \mathcal{B}_{\bm{y}_{\mathrm{b}}}^2: m' \neq m'', d_H\left( \mathcal{C}_k(m';\bm{z}_{k-1}'), \mathcal{C}_k(m'';\bm{z}_{k-1}') \right) \leq \Delta_k \} \nonumber
  \end{align}
  for some distance parameter $\Delta_k > 0$. Finally, the adversary constructs the set of all \textit{strongly confusable message} pairs $$\textset{SCM}_{\bm{y}_{\mathrm{b}}} = \bigcap_{k = T_{\mathrm{b}}+1}^{T+1} \bigcap_{\bm{z}_{k-1}' \in \mathcal{Z}_{k-1}'} \mathcal{D}_{k,\bm{z}_{k-1}'}.$$ In summary, the construction of $\textset{SCM}_{\bm{y}_{\mathrm{b}}}$ ensures the following two properties for any $(m',m'') \in \textset{SCM}_{\bm{y}_{\mathrm{b}}}$:
  \begin{enumerate}
  \item $m'$ and $m''$ are equally likely to be Alice's message given Bob's view of the first $b$ received symbols
  \item The remaining $n-b$ symbols of the codewords corresponding to $m'$ and $m''$ are close for \textit{all} possible feedback sequences, i.e., for any $\bm{z}'_T \in \mathcal{Z}'_{T}$
  \begin{equation} \label{eq:prop_2_Dk}
  d_H \left( \mathcal{C}^{(\mathrm{p})}(m';\bm{z}'_T), \mathcal{C}^{(\mathrm{p})}(m'';\bm{z}'_T) \right) \leq \sum_{k = T_{\mathrm{b}}+1}^{T+1} \Delta_k.
  \end{equation}
  \end{enumerate}
  We remark that property 2 greatly simplifies analysis of the ``babble-and-push'' attack by allowing us to treat the code $\Psi$ as if it was a code without feedback. In particular, since (\ref{eq:prop_2_Dk}) holds for all $\bm{z}'_T \in \mathcal{Z}'_{T}$, one can virtually ignore the presence of feedback when analyzing the closeness of codewords corresponding to two messages.

  In the analysis, we choose $\{ \Delta_k \}$ such that for any $(m',m'') \in \textset{SCM}_{\bm{y}_{\mathrm{b}}}$ and any $\bm{z}'_{T} \in \mathcal{Z}'_{T}$,
  \begin{align}
   d_H\big(\mathcal{C}^{(p)}(m';\bm{z}_T'),\mathcal{C}^{(p)}(m'';\bm{z}_T') \big) \leq \Delta \triangleq 2(p-\bar{p}) + rn - \frac{\epsilon_R n}{16 T}. \nonumber 
  \end{align}
  Given this choice of $\Delta_k$, we show that with probability over the ``babble'' attack and Alice's message distribution, Bob's received sequence $\bm{y}_{\mathrm{b}}$ is a sequence $\bm{y}^*_{\mathrm{b}} \in \mathcal{Y}^b$ such that the set $\textset{SCM}_{\bm{y}^*_b}$ is non-empty (c.f. Corollary \ref{thm:pos_p_new}). It immediately follows that there exists a message pair $(m_1,m_2) \in \textset{SCM}_{\bm{y}^*_b}$ and there exists a received sequence $\bm{y}^*_{\mathrm{p}} \in \mathcal{Y}^{n-b}$ that a) has at most $rn$ erasure symbols and b) is \textit{sufficiently close} to the codewords corresponding to message $m_1$ and $m_2$, i.e., 
  \begin{equation} \label{eq:max_dp}
  \max_{i \in \{ 1,2\}}d_H\left(\mathcal{C}^{(p)}(m_i;z_1^*,\ldots,z_{T_{\mathrm{b}}}^*),\bm{y}_{\mathrm{p}}^* \right)\leq (p-\bar{p})n - \frac{n \epsilon_R}{16}
  \end{equation}
  where $(z_1^*,\ldots,z_T^*) = f((\bm{y}_{\mathrm{b}}^*,\bm{y}_{\mathrm{p}}^*))$. Given the existence of the above $\bm{y}^*_{\mathrm{b}}$ and $\bm{y}^{*}_{\mathrm{p}}$, the adversary performs the following ``push'' attack.
  
  \textbf{``Push'' Attack:} Let $\bm{x}_{\mathrm{p}} = (x_{b+1}, \ldots, x_n)$ and $\bm{y}_{\mathrm{p}}=(y_{b+1},\ldots,y_n)$ denote Alice's transmission and Bob's received sequence, respectively, during the ``push'' attack. The adversary's push attack is as follows. If either $\bm{y}_{\mathrm{b}} \neq \bm{y}_{\mathrm{b}}^*$ or $m \not\in \{m_1,m_2\}$, then the adversary takes no further action and $\bm{y}_{\mathrm{p}} = \bm{x}_{\mathrm{p}}$. Otherwise, if $\bm{y}_{\mathrm{b}} = \bm{y}_{\mathrm{b}}^*$ and $m \in \{m_1,m_2\}$, then the adversary chooses $\adv \in \mathcal{ADV}$ to be the mapping which outputs $\bm{y}_{\mathrm{p}} = \bm{y}_{\mathrm{p}}^*$ (call this mapping $\adv^*$). We remark that such a mapping exists following (\ref{eq:max_dp}) and the assumption that the adversary can induce up to $pn$ symbol errors and $rn$ symbol erasures. In summary, with positive probability over the ``babble'' attack and message distribution, Alice sends $m \in \{m_1,m_2\}$ and
  \begin{equation} \nonumber
  \bm{y}^* = \adv^*(\mathcal{C}(m_1;z_1^*,\ldots,z_T^*)) = \adv^*(\mathcal{C}(m_2;z_1^*,\ldots,z_T^*))
  \end{equation}
  where $\bm{y}^* = (\bm{y}_{\mathrm{b}}^*,\bm{y}_{\mathrm{p}}^*)$ and $(z_1^*,\ldots,z_T^*) = f(\bm{y}^*)$. Hence, $(m_1,m_2)$ is a confusable message pair.

  \subsection{Analysis of ``Babble-and-Push'' Attack} 

  Let $M$, $\bm{X}_{\mathrm{b}}$ and $\bm{Y}_{\mathrm{b}}$ denote the random variables corresponding to Alice's message, Alice's first $b$ codeword symbols, and Bob's first $b$ received symbols, respectively. Assume that $M$ is uniformly distributed in the message set $\mathcal{M} = [q^{Rn}]$. Denote the Shannon entropy of $M$ conditioned on $\bm{Y}_{\mathrm{b}}$ as 
  \begin{equation} \nonumber
  H(M|\bm{Y}_{\mathrm{b}}) =  - \sum_{m \in \mathcal{M}} \sum_{\bm{y}_{\mathrm{b}} \in \mathcal{Q}^b} p(m,\bm{y}_{\mathrm{b}}) \log_2 p(m|\bm{y}_{\mathrm{b}})
  \end{equation}
  and let $\mathcal{F} = \{ \bm{Y}_{\mathrm{b}} \in \{ \bm{y}_{\mathrm{b}} \in \mathcal{Y}^b: H(M|\bm{Y}_{\mathrm{b}} = \bm{y}_{\mathrm{b}}) \geq \frac{n \epsilon_R}{4} \}\}$ be the event that Bob has uncertainty in message $M$ after observing $\bm{Y}_{\mathrm{b}}$.
  \begin{lemma} \label{thm:F_0}
  $\mathbb{P}\left( \mathcal{F} \right) > \frac{\epsilon_R}{4}$.
  \end{lemma}

  \begin{proof}[Proof of Lemma \ref{thm:F_0}]
  The proof is identical to the proof of \cite[Claim A.2]{Chen2019} and studies lower bounds on the entropy $H(M|\bm{Y}_{\mathrm{b}})$.
  \begin{align}
  & H(M|\bm{Y}_{\mathrm{b}}) \nonumber \\
  & = H(M) - I(M;\bm{Y}_{\mathrm{b}}) \nonumber \\
  & \geq H(M) - I(\bm{X}_{\mathrm{b}};\bm{Y}_{\mathrm{b}}) \label{eq:F_0_0} \\
  & \geq H(M) - b \left( 1- H_q \left( \frac{\bar{p}n}{b} \right) \right) \label{eq:F_0_0_1} \\
  & = n\left(\alpha(\bar{p}) \left(1 - H_q\left(\frac{\bar{p}}{\alpha(\bar{p})}\right) \right) + \epsilon_R \right) -  b \left( 1- H_q \left( \frac{\bar{p}n}{b} \right) \right) \label{eq:F_0_2} \\
  & = \frac{n \epsilon_R}{2} + n \left(\alpha(\bar{p})+\frac{\epsilon_R}{2} \right)H_q \left( \frac{\bar{p}}{\alpha(\bar{p})+\epsilon_R/2} \right) \nonumber \\
  & \hspace{13em} - n \alpha(\bar{p}) H_q \left( \frac{\bar{p}}{\alpha(\bar{p})} \right) \label{eq:F_0_3} \\
  & \geq \frac{n \epsilon_R}{2} \label{eq:F_0_4}
  \end{align}
  where (\ref{eq:F_0_0}) follows from the data processing inequality, (\ref{eq:F_0_0_1}) follows from the adversary's ``babble'' attack and the fact that feedback does not increase capacity, (\ref{eq:F_0_2}) follows from the assignment $R = \alpha(\bar{p})\left(1-H_q\left(\frac{\bar{p}}{\alpha(\bar{p})}\right) \right) + \epsilon_R$ where $\bar{p}$ is defined in (\ref{eq:p_bar_ass}), (\ref{eq:F_0_3}) follows from substituting $b = n(\alpha(\bar{p})+\epsilon_R/2)$, and (\ref{eq:F_0_4}) follows from the fact that the function $g(x) = x H_q \left(\frac{\bar{p}}{x}\right)$ is increasing in $x$.

  We note that for all $\bm{y}_{\mathrm{b}} \in \mathcal{Q}^b$ the entropy $H(M|\bm{Y}_{\mathrm{b}} = \bm{y}_{\mathrm{b}})$ is bounded above by $Rn$ and, thus, the quantity $Rn - H(M|\bm{Y}_{\mathrm{b}} = \bm{y}_{\mathrm{b}})$ positive. In turn, by Markov's inequality,
  \begin{align}
  \mathbb{P}\left(Rn - H(M|\bm{Y}_{\mathrm{b}}=\bm{y}_{\mathrm{b}}) > Rn - \frac{\epsilon_R n}{4} \right) &< \frac{nR - H(M|\bm{Y}_{\mathrm{b}})}{Rn - \frac{n\epsilon_R}{4}} \nonumber \\
  & \leq \frac{R-\frac{\epsilon_R}{2}}{R-\frac{\epsilon_R}{4}} \label{eq:F_0_5}
  \end{align}
  where the last inequality follows from (\ref{eq:F_0_4}). Hence,
  \begin{align}
  \mathbb{P}\left( \mathcal{F} \right) &= \mathbb{P}\left( H(M|\bm{Y}_{\mathrm{b}} = \bm{y}_{\mathrm{b}}) \geq \frac{n \epsilon_R}{4} \right) \nonumber \\
  &= 1-\mathbb{P}\left( Rn - H(M|\bm{Y}_{\mathrm{b}}=\bm{y}_{\mathrm{b}}) > Rn - \frac{n \epsilon_R}{4} \right) \nonumber \\
  & > 1 - \frac{R-\epsilon_R/2}{R-\epsilon_R/4} \label{eq:F_0_6} \\
  & \geq \frac{\epsilon_R}{4}. \nonumber
  \end{align}
  where (\ref{eq:F_0_6}) follows from (\ref{eq:F_0_5}). 
  \end{proof}

  \begin{corollaryLemma} \label{thm:size_By}
  Conditioned on event $\mathcal{F}$,  the number of messages in $\mathcal{B}_{\bm{y}_{\mathrm{b}}}$ is at least $q^{\frac{n \epsilon_R}{4} }$.
  \end{corollaryLemma}

  We now set the distance paramter $\Delta_k$. Let $\beta_k = \frac{t_k-t_{k-1}}{n-b}\in (0,1]$ denote the ratio of the number of symbols in sub-codeword $\mathcal{C}_k$ and the number of remaining symbols $n-b$ after the ``babble'' attack. In turn, for $k \in [T_{\mathrm{b}}+1,T+1]$, set 
  \begin{equation} \label{eq:Dk_def}
  \Delta_k = 
  \begin{cases}
  2(p - \bar{p}) \beta_k n + r \beta_k n - \frac{ \epsilon_R \beta_k n}{8}, & \beta_k \geq \frac{\epsilon_R}{16 T} \\
  (n-b) \beta_k, & \beta_k < \frac{\epsilon_R}{16 T}.
  \end{cases}
  \end{equation}
  In turn, for any $(m',m'') \in \textset{SCM}_{\bm{y}_{\mathrm{b}}}$ and for any $\bm{z}'_T \in \mathcal{Z}'_T$, the codewords corresponding to $m'$ and $m''$ are close in their last $n-b$ symbols such that
  \begin{align}
  &d_H \left( \mathcal{C}^{(\mathrm{p})}(m';\bm{z}'_T), \mathcal{C}^{(\mathrm{p})}(m'';\bm{z}'_T) \right) \nonumber \\
  & \stackrel{(*)}{\leq} \sum_{k=T_{\mathrm{b}}+1}^T \Delta_k \nonumber \\
  & = \sum_{k=T_{\mathrm{b}}+1}^T \left( 2(p-\bar{p})\beta_k n + r \beta_k n - \frac{\epsilon_R \beta_k n}{8} \right) \mathds{1} \left\{\beta_k \geq \frac{\epsilon_R}{16T}\right\} + \sum_{k=T_{\mathrm{b}}+1}^T (n-b) \beta_k \mathds{1} \left\{ \beta_k < \frac{\epsilon_R}{16T} \right\} \nonumber \\
  & \leq \sum_{k=T_{\mathrm{b}}+1}^T \left( 2(p-\bar{p})\beta_k n + r \beta_k n - \frac{\epsilon_R \beta_k n}{8} \right) + \sum_{k=T_{\mathrm{b}}+1}^T (n-b) \frac{\epsilon_R}{16T} \nonumber \\
  &\leq \Delta \triangleq 2(p-\bar{p})n + rn - \frac{n \epsilon_R}{16} \nonumber
  \end{align}
  where (*) follows from (\ref{eq:prop_2_Dk}).

  A non-binary version of the Plotkin Bound is needed for the next lemma.
  \begin{lemma}[{Plotkin Bound \cite{Blake1976}}] \label{thm:Plotkin}
  A $q$-ary $(n,k,0)$-code $\Psi = (\mathcal{C},\phi)$ with minimum distance $d_{min} > (1-1/q)n$ must have a bounded number of codewords such that $|\mathcal{C}| \triangleq q^k \leq \frac{q d_{min}}{q d_{min} - (q-1)n}$.
  \end{lemma}

  \begin{lemma} \label{thm:E_y}
  Conditioned on event $\mathcal{F}$, the set of strongly confusable message pairs $\textset{SCM}_{\bm{y}_{\mathrm{b}}}$ is non-empty for large enough $n$.
  \end{lemma}

  \begin{proof}[Proof of Lemma \ref{thm:E_y}]
  Recall that the set of all feedback sequences that Bob may send in the ``push'' phase which are consistent with the feedback $z_1, \ldots,z_{T_{\mathrm{b}}}$ sent during the ``babble'' phase is
  \begin{equation} \label{eq:fb_enum}
  \bigcup_{k=T_{\mathrm{b}}+1}^{T+1} \mathcal{Z}'_{k-1} \triangleq \{z_1\} \times \ldots \times \{z_{T_{\mathrm{b}}}\} \times \bigcup_{k=T_{\mathrm{b}}+1}^{T+1} \mathcal{Z}^{k-1-T_{\mathrm{b}}}.
  \end{equation}
  Let $I$ denote the number of sequences in the set (\ref{eq:fb_enum}), and in turn, let $\bm{z}^{(1)}, \bm{z}^{(2)}, \ldots, \bm{z}^{(I)}$ be any enumeration of sequences in (\ref{eq:fb_enum}). Note that for $i \in [I]$, we have that $\bm{z}^{(i)} \in \mathcal{Z}'_{k_i-1}$ for some sub-codeword index $k_i \in [T_{\mathrm{b}}+1,T+1]$.
  
  For each $i \in [I]$, we construct a graph to study the distance between a particular subset of $k_i^\text{th}$ sub-codewords corresponding to feedback sequence $\bm{z}^{(i)}$. For $i = 1,2, \ldots, I$, consider a simple undirected graph $\mathcal{G}_{i} = (\mathcal{V}_i,\mathcal{E}_i)$ with a vertex set $\mathcal{V}_{i}$ consisting of some subset of sub-codewords $\{ \mathcal{C}_{k_i}(m';\bm{z}^{(i)}): m' \in \mathcal{B}_{\bm{y}_{\mathrm{b}}} \}$ (we provide a detailed construction of $\mathcal{V}_{i}$ shortly). Two distinct sub-codewords $\bm{x}'$ and $\bm{x}''$ in $\mathcal{V}_{i}$ are connected by an edge if and only if $d_H(\bm{x}',\bm{x}'') \leq \Delta_{k_i}$.
  
  To describe the construction of $\mathcal{V}_{i}$, we first define a maximum clique and maximum independent set of $\mathcal{G}_{i}$. A \textit{maximum clique} $\mathcal{K}_{i}$ of graph $\mathcal{G}_{i}$ corresponds to a largest subset of sub-codewords in $\mathcal{V}_{i}$ such that every 2 sub-codewords are within Hamming distance $\Delta_{k_i}$. We let $\mathcal{K}_i$ be some maximum clique in case two or more such cliques may exists. A \textit{maximum independent set} $\mathcal{I}_{i}$ of graph $\mathcal{G}_{i}$ corresponds to a largest subset of sub-codewords in $\mathcal{V}_{i}$ such that every 2 sub-codewords are not within Hamming distance $\Delta_{k_i}$. Recall that Plotkin's bound (Lemma \ref{thm:Plotkin}) provides an upper bound on the number of (sub-) codewords that are not within Hamming distance $\Delta_{k_i}$ from each other. It is easy to verify that $\Delta_{k_i}$ satisfies the Plotkin's bound condition $\Delta_{k_i} > (1-1/q)(n-b)\beta_{k_i}$. In turn, the size of a maximum independent set is bounded such that $$|\mathcal{I}_{i}| \leq \frac{q \Delta_{k_i}}{q \Delta_{k_i} - (q-1)(n-b)\beta_{k_i}}.$$ In turn, by substituting our above choice of $\Delta_{k_i}$ (c.f. (\ref{eq:Dk_def})) and the assignment of $b$, the bound on $|\mathcal{I}_i|$ can be simplified to $|\mathcal{I}_{i}| \leq N \triangleq \max\{\frac{8(p-\bar{p})+8r}{3\epsilon_R},q \}$ where we note that the upper bound $N$ is constant in $n$.

  We now construct the vertex sets $\mathcal{V}_1,\mathcal{V}_2,\ldots,\mathcal{V}_I$ in a recursive manner such that the sub-codewords in $\mathcal{V}_i$ correspond to the messages of sub-codewords in the maximum clique $\mathcal{K}_{i-1}$. As the base case ($i=1$), we define 
  \begin{equation} \nonumber
  \mathcal{V}_{1} \triangleq \{ \mathcal{C}_{k_1}(m';\bm{z}^{(1)}): m' \in \mathcal{B}_{\bm{y}_{\mathrm{b}}} \}.
  \end{equation}
  For $i = 2,3, \ldots, I$, the vertex set
  \begin{equation} \nonumber
  \mathcal{V}_{i} =
  \{\mathcal{C}_{k_i}(m';\bm{z}^{(i)}): m' \in \mathcal{K}_{i-1} \}
  \end{equation}
  where $m' \in \mathcal{K}_{i-1}$ denotes the message $m'$ corresponding to the sub-codeword $\mathcal{C}_{k_{i-1}}(m';\bm{z}^{(i-1)}) \in \mathcal{K}_{i-1}$. Thus, all messages corresponding to sub-codewords in $\mathcal{V}_i$ have sub-codewords in $\mathcal{V}_{i-1}$ that are pairwise close. In turn, if two unique messages $m'$ and $m''$ have corresponding sub-codewords in $\mathcal{V}_I$, then $m'$ and $m''$ have corresponding sub-codewords that are pairwise close in $\mathcal{V}_i$ for \textit{all} $i \in \{1,\ldots, I\}$ and, in turn, $(m',m'') \in \textset{SCM}_{\bm{y}_{\mathrm{b}}}$. Thus, to prove Lemma \ref{thm:E_y}, it is sufficient to show that the vertex set $\mathcal{V}_I$ contains at least two sub-codewords. This is equivalent to showing that maximum clique $\mathcal{K}_{I-1}$ contains at least two sub-codewords.

  We show the above sufficient condition by lower bounding the size of the cliques $\{\mathcal{K}_i\}_{i \in [I-1]}$. We introduce the \textit{Ramsey number} $\textrm{R}(|\mathcal{K}|,|\mathcal{I}|)$ which is the smallest integer such that every simple undirected graph of size $\textrm{R}(|\mathcal{K}|,|\mathcal{I}|)$ has a clique of size $|\mathcal{K}|$ or an independent set of size $|\mathcal{I}|$. Recall that the size of a maximum independent set $\mathcal{I}_i$ is at most $N$. Thus, for $K\geq1$, if the size of the vertex set $\mathcal{V}_i$ is at least $\textrm{R}\left(K,N+1 \right)$, then the size of the maximum clique $\mathcal{K}_i$ is at least $K$.

  We now prove by induction that for large enough $n$, for all $i \in [I-1]$ the size of $\mathcal{K}_i$ is at least $L_i (q^n)$ where $L_i (q^n)$ denotes the composition $\log_q \circ \log_q \circ \cdots \circ \log_q(q^n)$ of $i$ number of logarithms. (\textit{Base Case}) Recall that the event $\mathcal{F}$ occurs by assumption. Then by Corollary \ref{thm:size_By}, the size of $\mathcal{V}_{1}$ is at least $q^{\frac{n \epsilon_R}{4}}$. It follows that the size of $\mathcal{K}_1$ is at least $L_1(q^n) = n$ if $\textrm{R}\left(n, N+1\right) \leq q^{\frac{n \epsilon_R}{4}}$. Indeed, for large enough $n$
  \begin{align}
  \textrm{R}\left(n, N+1\right) & \stackrel{(a)}{\leq} { n + N - 1 \choose N } \nonumber \\
  &\stackrel{(b)}{\leq} 2^{\left(n + N - 1\right) H_2 \left(\frac{N}{n + N - 1} \right)} \nonumber  \\
  & \stackrel{(c)}{\leq} q^{\frac{n \epsilon_R}{4}} \leq |\mathcal{V}_1| \nonumber
  \end{align}
 where (a) follows from the bound that for any $|\mathcal{K}|, |\mathcal{I}| \geq 1$, $\textrm{R}\left( |\mathcal{K}|, |\mathcal{I}| \right)$ is at most ${|\mathcal{I}|+|\mathcal{K}|-2 \choose |\mathcal{I}|-1}$ \cite{Greenwood1955}, (b) follows from the bound ${n \choose k} \leq 2^{n H_2(k/n)}$ for $k \leq n/2$, and (c) follows for large enough $n$ from the fact that $N$ is constant in $n$, and thus, $H_2(\frac{N}{n+N-1})$ tends to $0$ as $n$ tends to infinity. In conclusion, $|\mathcal{K}_1| \geq n$. Base case done. (\textit{Induction}) Let $i \in \{2,\ldots,I-2\}$ and assume that the size of $\mathcal{K}_{i-1}$ is at least $L_{i-1}(q^n)$. Following the construction of $\mathcal{V}_i$, the size of $\mathcal{V}_i$ is equal to the size of $\mathcal{K}_{i-1}$. Similar to the argument made in the base case, we have that
 $|\mathcal{K}_i| \geq L_i(q^n)$ if
  $\textrm{R}\left(L_i(q^n), N+1\right) \leq |\mathcal{V}_i| \triangleq |\mathcal{K}_{i-1}|$. Indeed, for large enough $n$, 
  \begin{align}
  \textrm{R}\left(L_i(q^n), N+1\right) & \leq { L_i(q^n) + N - 1 \choose N } \nonumber \\
  & \leq 2^{(L_i(q^n)+N-1) H_2 \left( \frac{N}{L_i(q^n)+N-1} \right)} \nonumber \\
  %\begin{align}
  %&\textrm{R}\left(L_i(q^n), \frac{8(p-\bar{p})+8r}{3 \epsilon_R}+1\right) \leq { n + \frac{8(p-\bar{p})+8r}{3 \epsilon_R} - 1 \choose \frac{8(p-\bar{p})+8r}{3 \epsilon_R} } \nonumber
  %\\ \leq 2^{\left(L_i(q^n) + \frac{8(p-\bar{p})+8r}{3 \epsilon_R} - 1\right) H_2 \left(\frac{\frac{8(p-\bar{p})+8r}{3 \epsilon_R}}{L_i(q^n) + \frac{8(p-\bar{p})+8r}{3 \epsilon_R} - 1} \right)} \nonumber \\
  %&\leq L_{i-1}(q^n) \leq |\mathcal{K}_{i-1}| \triangleq |\mathcal{V}_i| \nonumber
  %\end{align}
  & \leq L_{i-1}(q^n) \stackrel{(d)}{\leq} |\mathcal{K}_{i-1}| \triangleq |\mathcal{V}_i| \nonumber
  \end{align}
  where (d) follows from our induction assumption. Since the number of feedback bits $B+1$ is a constant in $n$, it follows that $I$ is a constant, and thus, for large enough $n$, $|\mathcal{K}_i| \geq L_i(q^n)$ for all $i \in [I-1]$. Furthermore, $L_{I-1}(q^n) \geq 2$ for large enough $n$. Thus, for large enough $n$, $\mathcal{K}_{I-1}$ contains at least two sub-codewords.
  \end{proof}

   Let blocklength $n$ be large enough such that Lemma \ref{thm:E_y} holds. Let $\bm{y}_{\mathrm{b}}^* \in \mathcal{Y}^b$ be any received sequence such that $\bm{Y}_{\mathrm{b}} = \bm{y}_{\mathrm{b}}^*$ with positive probability and $H(M|\bm{Y}_{\mathrm{b}} = \bm{y}_{\mathrm{b}}^*) \geq \frac{n \epsilon_R}{n}$. To complete the proof, we note the following corollary of Lemma \ref{thm:E_y}.

  \begin{corollaryLemma} \label{thm:pos_p_new}
  With positive probability over the ``babble'' attack and choice of Alice's message, $\bm{Y}_{\mathrm{b}} = \bm{y}_{\mathrm{b}}^*$ and there exists messages $m_1$ and $m_2$ such that $(m_1,m_2) \in \textset{SCM}_{\bm{y}_{\mathrm{b}}^*}$. Hence $(m_1,m_2)$ is a confusable pair.
  \end{corollaryLemma}

  \section{Conclusion} \label{sec:conclusion}

  This paper has initiated the study of the $q$-ary adversarial error-erasure channel with $O(1)$-bit feedback. As our main result, we provided a tight characterization of the zero-error capacity $C^{O(1)}_q(p,r)$ as a function the error fraction $p \in [0,1]$, erasure fraction $r \in [0,1]$ and input alphabet $q \geq 2$. In turn, this result a tight characterization of the zero-error capacity $C_q^{\mathrm{full}}(p,0)$ of the $q$-ary adversarial error channel with full noiseless feedback for sufficiently small error fraction $p$. 
  
  Following our initial study, some important questions remain open for future research. Recall that in the achievability analysis, our coding scheme uses a number of feedback bits $B$ that is constant in the blocklength $n$ but grows arbitrarily large as the coding rate $R$ tends to the capacity $C^{O(1)}_q(p,r)$. It remains an open question whether the condition that $B$ tends to $\infty$ as $R$ tends to $C^{O(1)}_q(p,r)$ is necessary or just a artifact of our coding scheme. Related to this question is the following: ``For a fixed constant $B \geq 0$, what is the zero-error capacity $C^B_q(p,r)$''? Answering this question in full appears to be significantly harder than fully characterizing $C^{O(1)}_q(p,r)$.

  \appendices

  \section{Proof of Equation (\ref{eq:cap_error})} \label{sec:main_cor1_proof}

  Recall that for a function $$g_p(\bar{p}) \triangleq \alpha(\bar{p})\left(1- H_q\left(\frac{\bar{p}}{\alpha(\bar{p})}\right)\right)$$ where $\alpha(\bar{p}) \triangleq 1 - \frac{2q}{q-1}(p-\bar{p})$, Theorem \ref{thm:cap_bf} states that $C^{O(1)}_q(p,0) = \min\limits_{\bar{p} \in [0,p]} g_p(\bar{p})$. Using this notation, we rewrite (\ref{eq:cap_error}):
  \begin{equation} \label{eq:cap_error_2}
  C_q^{O(1)}(p,0) = 
  \begin{cases}
  g_p(p), & p \in [0,p^*] \\
  g_p(\bar{p}^\dag), & p \in (p^*,\frac{q-1}{2q}) \\
  0, & p \in [\frac{q-1}{2q},1].
  \end{cases}
  \end{equation}
  where $\bar{p}^\dag$ is the unique value in $[0,p]$ such that $\frac{\bar{p}^\dag}{\alpha(\bar{p}^\dag)} = p^*$. Therefore, to show that (\ref{eq:cap_error_2}) follows from Theorem \ref{thm:cap_bf}, it is sufficient to show that 
  \begin{equation} \label{eq:suf}
  \arg\min\limits_{\bar{p} \in [0,p]} g_p(\bar{p}) = \begin{cases} p, & p \in (0,p^*] \\
  \bar{p}^\dag, & p \in (p^*,\frac{q-1}{2q})\end{cases}
  \end{equation}
  We show that equation (\ref{eq:cap_error}) (equivalently, (\ref{eq:suf})) holds by showing that the following three conditions hold:
  \begin{enumerate}
  \item For all $p \in [0,\frac{q-1}{2q}]$, $g_p(\bar{p})$ is convex in $\bar{p}$. 
  \item $\frac{\partial g_p(p)}{\partial \bar{p}} \leq 0$ for all $p \in (0,p^*]$.
  \item $\frac{\partial g_p(\bar{p}^\dag)}{\partial \bar{p}} = 0$ for all $p \in (p^*,\frac{q-1}{2q})$.
  \end{enumerate}
  In terms of the optimization problem in (\ref{eq:suf}), we remark that condition 1 is the second order sufficient condition while condition 2 and 3 are first order necessary conditions.

  To show condition 1, we show that $\frac{\partial^2 g_p}{\partial \bar{p}^2} \geq 0$ for all $\bar{p} \in [0,p]$. For $p \in [0,\frac{q-1}{2q}]$, the first derivative of $g$ is 
  \begin{align}
  \frac{\partial g_p(\bar{p})}{\partial \bar{p}} = - \log_q(q-1) + \frac{2q}{q-1} +  \log_q\left(\frac{\bar{p}}{\alpha(\bar{p})} \right) + \frac{q+1}{q-1} \log_q \left(1 - \frac{\bar{p}}{\alpha(\bar{p})} \right), \label{eq:der_g}
  \end{align}
  where it follows that $\frac{\partial^2 g_p}{\partial \bar{p}^2} \geq 0$ if 
  $h(\bar{p})(1 - h(\bar{p}))^{\frac{q+1}{q-1}}$ is increasing in $\bar{p}$ where $h(\bar{p}) \triangleq \frac{\bar{p}}{\alpha(\bar{p})}$. We have that
  \begin{align}
  \frac{\partial}{\partial \bar{p}} (h(1-h)^\frac{q+1}{q-1}) = \frac{\partial h}{\partial \bar{p}}(1-h)^\frac{q+1}{q-1}\left(1- \frac{h}{1-h} \frac{q+1}{q-1}\right)\nonumber
  \end{align}
  is increasing in $\bar{p} \in [0,p]$ since $h(\bar{p})$ is increasing in $\bar{p}$ and, in turn, $\frac{h}{1-h} \frac{q+1}{q-1}$ is bounded above by $1$ for all $\bar{p} \in [0,p]$. Done.

  Next, we show that condition 2 holds. From (\ref{eq:der_g}), we have that
  $\frac{\partial g_p(p)}{\partial \bar{p}} = -\log_q(q-1) + \frac{2q}{q-1} + \log_q(p) + \frac{q+1}{q-1}\log_q(1-p)$, where we note that $\frac{\partial g_p(p)}{\partial \bar{p}}$ is increasing in $p$ over the interval $[0,1/2]$. It follows that $\frac{\partial g_p(p)}{\partial \bar{p}} \leq 0$ for $p \in (0,p^*]$ where $p^*$ is the unique value such that $\log_q(p^*) + \frac{q+1}{q-1}\log_q(1-p^*) = \log_q(q-1) - \frac{2q}{q-1}$.

  Lastly, we show that condition 3 holds. It is straightforward to verify that $$\bar{p}^\dag = \frac{1-\frac{2q}{q-1}p}{\frac{1}{p^*} - \frac{2q}{q-1}},$$ and in turn, $\bar{p}^\dag$ is well defined in $[0,p]$ for all $p \in [p^*,\frac{q-1}{2q}]$. Now, from (\ref{eq:der_g}), we have that
  $\frac{\partial g_p(\bar{p}^\dag)}{\partial \bar{p}} = -\log_q(q-1) + \frac{2q}{q-1} + \log_q(p^*) + \frac{q+1}{q-1}\log_q(1-p^*)$ which in turn is equal to $0$ following the definition of $p^*$.

  \section{Proof that $p^* \leq 1/q$} \label{sec:p_star_q}

  Recall that for integer $q \geq 2$, $p^*$ is defined as the unique value in $[0,\frac{q-1}{2q}]$ such that 
  \begin{equation} \label{eq:p_star_q_1}
  p^*(1-p^*)^{\frac{q+1}{q-1}} = (q-1)q^{-\frac{2q}{q-1}}.
  \end{equation} 
  We show that $p^* \leq \frac{1}{q}$ for all $q \geq 2$. 

 Let $q \geq 2$ be an integer. By a rearrangement of terms, (\ref{eq:p_star_q_1}) is equivalent to $(q-p^*q)^{\frac{q+1}{q-1}} = \frac{q-1}{p^*q}$. Define the function $g(x) = (q-x)^{\frac{q+1}{q-1}} - \frac{q-1}{x}$ for all $x \in [0,1]$, and note that $p^*$ is the unique value in $[0,\frac{q-1}{2q}]$ such that $g(p^*q) =0$. It is straightforward to verify that the function $g$ is increasing; in particular, the derivative $\frac{\partial g(x)}{\partial x} > 0$ for all $x \in [0,\frac{q-1}{2}]$. Thus, we can show that $p^* \leq 1/q$ if $g(1) > 0$. Indeed, $g(1) = (q-1)^{\frac{q+1}{q-1}} - (q-1) > 0$.

  \section{Proof of Corollary \ref{thm:lg_alphabet}} \label{sec:lg_alphabet_proof} 

  Let $p \in [0,1]$ and $r\in [0,1]$ such that $2p+r<1$. We show that $C^{O(1)}_q(p,r) = 1 - 2p - r - \Theta(1/q)$. We have that
  \begin{align}
  C^{O(1)}_q(p,r) &= \min\limits_{\bar{p} \in [0,p]} \left[ \alpha(\bar{p}) \left(1 - H_q\left( \frac{\bar{p}}{\alpha(\bar{p})}\right) \right) \right] \label{eq:cor_lim_1} \\
  & = \min\limits_{\bar{p} \in [0,p]} \left[ \alpha(\bar{p}) \left(1 - \frac{\bar{p}}{\alpha(\bar{p})} + o(1) \right) \right] \label{eq:cor_lim_2}
  \end{align}
  where (\ref{eq:cor_lim_1}) follows from Theorem  \ref{thm:cap_bf} for large enough $q$, and the $o(1)$ term in (\ref{eq:cor_lim_2}) tends to $0$ as $q$ tends to $\infty$. Note that the $o(1)$ term is $0$ when $\bar{p}=0$. Moreover, (\ref{eq:cor_lim_2}) is equal to
  \begin{align}
  &  \min\limits_{\bar{p} \in [0,p]} \left[ \alpha(\bar{p}) - \bar{p} + o(1) \right] \nonumber \\
  & = 1 - \frac{q(2p +r)}{q-1} + \min\limits_{\bar{p} \in [0,p]} \left[ \left(\frac{2q}{q-1}-1 \right)\bar{p} + o(1) \right] \label{eq:cor_lim_3} \\
  & = 1 - \frac{q(2p +r)}{q-1} \label{eq:cor_lim_4} \\
  & = 1-2p-r - \Theta\left(\frac{1}{q}\right) \nonumber
  \end{align}
  where (\ref{eq:cor_lim_3}) follows from substituting $\alpha(\bar{p}) = 1 - \frac{2q}{q-1} (p - \bar{p}) - \frac{q}{q-1}r$, (\ref{eq:cor_lim_4}) follows from the observation that $\bar{p}=0$ is a minimizer for large enough $q$.

  For completeness, we also show that $C^{\mathrm{full}}_q(p,r) = 1 - 2p -r - \Theta\left(\frac{1}{q \log q} \right)$. For large enough $q$,
  \begin{align}
  C^{\mathrm{full}}_q(p,r) &= (1-2p) \log_q (q-1) - r \label{eq:cor_lim_01} \\
  & = (1-2p) \frac{\ln(q-1)}{\ln q} - r \label{eq:cor_lim_02} \\
  & = 1-2p-r + (1-2p) \frac{\ln(1-\frac{1}{q})}{\ln q} \label{eq:cor_lim_03}
  \end{align}
  where (\ref{eq:cor_lim_01}) is a result of \cite{Ahlswede2006}, (\ref{eq:cor_lim_02}) follows from a logarithm change-of-base and (\ref{eq:cor_lim_03}) follows from a rearrangement of terms. By the elementary bounds $\frac{-x}{1-x} \leq \ln( 1-x ) \leq -x$ for $x<1$, we have that $\ln(1-1/q) = -\Theta(1/q)$ and, in turn, $C^{\mathrm{full}}_q(p,r) = 1-2p-r -\Theta(\frac{1}{q \log q})$.

  \section{Proof of Lemma \ref{thm:LD}} \label{sec:LD_proof}
  
   We restate the probabilistic proof from {\cite[Claim~III.14]{Chen2019}} which is in the spirit of \cite[Theorem~10.3]{Guruswami2001}. Let $t \in \mathcal{T}$ and let $\lambda_t \in [0,rn]$ such that $t \geq t_{\mathrm{min}}$ and note that our definition of reference trajectory $\hat{p}_t$ ensures that $(t-\lambda_{t})(1-H_q(\hat{p}_t)) - \epsilon_Ln = Rn$. Furthermore, for a feedback sequence $(z_1,\ldots,z_{k-1}) \in \mathcal{Z}^{k-1}$ let $C_{\mathrm{prefix},k}(\cdot;z_1,\ldots,z_{k-1})$ denote the encoding function prefix $\mathcal{C}_1(\cdot) \circ \mathcal{C}_2(\cdot;z_1) \circ \cdots \circ \mathcal{C}_{k}(\cdot;z_{k-1})$. We evaluate the probability that the encoding function prefix $C_{\mathrm{prefix},k}$ is \textit{not} list-decodable for up to $(t-\lambda_t)\hat{p}_t$ symbol errors and $\lambda_t$ symbol erasures with a given list size $L \geq 1$. 
   
   Let $\bm{y}_{\mathrm{prefix},k} \in \{ \mathcal{Q}\cup \{?\} \}^t$ be a sequence with exactly $\lambda_t$ symbol erasures and let $(z_1,\ldots,z_{k-1},\cdot) = f(\bm{y}_{\mathrm{prefix},k},\cdot)$. We count the number of codeword prefixes in the set $\{C_{\mathrm{prefix},k}(m'; z_1,\ldots,z_{T}): m'\in \mathcal{M}\}$ within Hamming distance $(t-\lambda_{t})\hat{p}_t$ of the prefix $\bm{y}_{\mathrm{prefix},k}$. The probability that the codeword prefix of message $m' \in [q^{Rn}]$ is within a distance $(t-\lambda_t)\hat{p}_t$ of $\bm{y}_{\mathrm{prefix},k}$ is $|\mathcal{B}_{(t-\lambda_{t})\hat{p}_t}| q^{-(t-\lambda_t)}$ where $|\mathcal{B}_{(t-\lambda_{t})\hat{p}_t}|$ is the number of sequences in $\mathcal{Q}^{t-\lambda_t}$ contained in a Hamming ball of radius $(t-\lambda_t)\hat{p}_t$ centered at the unerased symbols of $\bm{y}_{\mathrm{prefix},k}$. In turn, following a simple union bound, the probability that there exists $L+1$ messages within a distance $(t-\lambda_t)\hat{p}_t$ of $\bm{y}_{\mathrm{prefix},k}$ is at most 
   \begin{equation} \nonumber
   {q^{Rn} \choose L+1}|\mathcal{B}_{(t-\lambda_t)\hat{p}_t}|^{L+1} q^{-(t-\lambda_t)(L+1)}.
   \end{equation}
   Therefore, the probability there is \textit{some} sequence $\bm{y}_{\mathrm{prefix},k} \in \{ \mathcal{Q}\cup \{?\} \}^t$ with exactly $\lambda_t$ symbol erasures and corresponding feedback tuple $(z_1,\ldots,z_{k-1})$ such that there exists $L+1$ messages corresponding to codeword prefix in $\{C_{\mathrm{prefix},k}(m';z_1,\ldots,z_{k-1}):m'\in \mathcal{M} \}$ that are each within a distance $(t-\lambda_t)\hat{p}_t$ of $\bm{y}_{\mathrm{prefix},k}$ is at most
   \begin{equation} \label{eq:LD_2}
   {t \choose \lambda_t}q^{(t-\lambda_t)} {q^{Rn} \choose L+1}|\mathcal{B}_{(t-\lambda_t)\hat{p}_t}|^{L+1} q^{-(t-\lambda_t)(L+1)}.
   \end{equation}

   To further bound (\ref{eq:LD_2}), we note the following bounds. First, ${q^{Rn} \choose L+1} \leq q^{Rn(L+1)}$ and ${t \choose \lambda_t} \leq 2^{t}$. Second, it follows from Lemma \ref{thm:bin_ub} that
   \begin{equation} \nonumber
   |\mathcal{B}_{(t-\lambda_t)\hat{p}_t}| = \sum_{i=0}^{(t-\lambda_t)\hat{p}_t} {t- \lambda_t \choose i} (q-1)^i \leq q^{(t-\lambda_t)H_q(\hat{p}_t)}.
   \end{equation}
   It follows that (\ref{eq:LD_2}) is bounded above by 
   \begin{align} 
   &  q^{ t\log_q2 + (t-\lambda_t) + (Rn - (t-\lambda_t)(1-H_q(\hat{p}_t)))(L+1) } \nonumber \\
   & = q^{t\log_q2 + (t-\lambda_t) -\epsilon_Ln(L+1)} \label{eq:LD_3}
   \end{align}
   where (\ref{eq:LD_3}) follows from the definition of $\hat{p}_t$. Then (\ref{eq:LD_3}) is equal to $2^{-n}$ when $\log_q2t + (t-\lambda_t) - \epsilon_L n (L+1) = -n$. Solving for $L$ gives 
   \begin{equation} \label{eq:LD_3_5}
   L = \frac{n + t\log_q2 + (t-\lambda_t)}{\epsilon_L n} - 1
   \end{equation}
   and, in turn, using the trivial bounds $t-\lambda_t \leq t \leq n$, the RHS of (\ref{eq:LD_3_5}) is bounded above by
   \begin{equation} \label{eq:LD_4}
   \frac{1 + \log_q2 + 1}{\epsilon_L} - 1.
   \end{equation}
   Therefore, (\ref{eq:LD_2}) is bounded above by $q^{-n}$ for any $L$ greater than (\ref{eq:LD_4}).

\bibliographystyle{IEEEtran}
\bibliography{refs}

% \fi
\IEEEpeerreviewmaketitle

\end{document}